\documentclass[journal]{IEEEtranTCOM}

\usepackage{amsmath, amsthm, amsfonts, amssymb, amsbsy}

\usepackage{graphicx}
\usepackage{stfloats}

\usepackage{pstricks}
\usepackage{algorithm}
\usepackage{multirow}
\usepackage{algorithmic}
\usepackage{bm}
\usepackage{cite}
\usepackage{url}
\usepackage{enumerate}
\usepackage{float}
\usepackage{amsmath}
\usepackage{amssymb}
\usepackage{mathrsfs}
\usepackage{lipsum} 
\usepackage{amsmath,epsfig,amssymb,subfigure,bm,dsfont}
\usepackage{multicol} 
\usepackage{booktabs}
\usepackage{threeparttable}
\usepackage{amsmath}
\usepackage{threeparttable}
\newtheorem{theorem}{Theorem}
\newtheorem{assumption}{Assumption}
\newtheorem{proposition}{Proposition}

\newtheorem{corollary}{Corollary}
\newtheorem{lemma}{Lemma}
\newtheorem{example}{Example}

\newtheorem{remark}[theorem]{Remark}

\newcommand{\twotriangle}{\hfill $\bigtriangleup \bigtriangleup$  }
\newcommand{\eax}{\twotriangle  \end{example}}
\newcommand\bim{\begin{itemize}}
\newcommand\eim{\end{itemize}}

\begin{document}

\title{
Ambiguity Function Analysis of AFDM Under Pulse-Shaped Random ISAC Signaling 
}

\author{
	Yuanhan Ni, {\it Member, IEEE}, Fan Liu, {\it Senior Member, IEEE}, Haoran Yin, Yanqun Tang, Zulin Wang, {\it Member, IEEE}
	\thanks{Part of this paper has been accepted by the 2025 IEEE International Conference on communication Workshops (ICC Workshops)\cite{ni2025AFDM}. (Corresponding author: Fan Liu and Zulin Wang.)}
	\thanks{Y. Ni and Z. Wang are with the School of Electronic and Information Engineering, Beihang University, Beijing 100191, China (e-mail: yuanhanni@buaa.edu.cn; wzulin@buaa.edu.cn).}
	\thanks{F. Liu is with the National Mobile communication Research Laboratory, Southeast University, Nanjing 210096, China (e-mail: fan.liu@seu.edu.cn).}
	\thanks{H. Yin and Y. Tang is with the School of Electronics and Communication Engineering, Sun Yat-sen University, Shenzhen 518107, China (e-mail: yinhr6@mail2.sysu.edu.cn; tangyq8@mail.sysu.edu.cn).}
}

\maketitle

\begin{abstract}
This paper investigates the ambiguity function (AF) of the emerging affine frequency division multiplexing (AFDM) waveform for Integrated Sensing and Communication (ISAC) signaling under a pulse shaping regime. Specifically, we first derive the closed-form expression of the average squared discrete period AF (DPAF) for AFDM waveform without pulse shaping, revealing that the AF depends on the parameter $c_1$ and the kurtosis of random communication data, while being independent of the parameter $c_2$. As a step further, we conduct a comprehensive analysis on the AFs of various waveforms, including AFDM, orthogonal frequency division multiplexing (OFDM) and orthogonal chirp-division multiplexing (OCDM). Our results indicate that all three waveforms exhibit the same number of regular depressions in the sidelobes of their AFs, which incurs performance loss for detecting and estimating weak targets. However, the AFDM waveform can flexibly control the positions of depressions by adjusting the parameter $c_1$, which motivates a novel design approach of the AFDM parameters to mitigate the adverse impact of depressions of the strong target on the weak target. Furthermore, a closed-form expression of the average squared DPAF for pulse-shaped random AFDM waveform is derived, which demonstrates that the pulse shaping filter generates the shaped mainlobe along the delay axis and the rapid roll-off sidelobes along the Doppler axis. Numerical results verify the effectiveness of our theoretical analysis and proposed design methodology for the AFDM modulation.
\end{abstract}


\section{Introduction}

The integrated sensing and communication (ISAC) has been recognized as one of the six key usage scenarios in the 6G vision by the International Telecommunication Union (ITU), due to its ability to simultaneously deliver useful information to communication users and sense the surrounding environment \cite{liu2020joint}. To implement ISAC, waveform design plays a vital role, as it is required to accomplish dual tasks in a single transmission over shared time-frequency resources. ISAC waveforms are generally designed following three methodologies: \romannumeral 1) sensing-centric design, \romannumeral 2) communication-centric design, and \romannumeral 3) joint design \cite{liu2022integrated}. The sensing-centric ISAC waveform design aims to insert communication data into classical radar waveforms. For example, minimum shift keying (MSK) or continuous phase modulation (CPM) modulations can naturally be combined with linear frequency modulation (LFM) carriers to conceive an ISAC waveform \cite{zhang2017modified}. Nevertheless, the resulting spectral efficiency of such waveforms is critically dependent on the pulse repetition frequency (PRF) of the radar and therefore struggles to support modern communication applications that require high data rates.


The communication-centric ISAC waveform design, on the other hand, aims to enable sensing functionality by reusing existing communication waveforms, such as single-carrier (SC) modulations and multi-carrier modulations. Specifically, due to their high spectral efficiency, multi-carrier modulations have been widely studied as candidate waveforms for ISAC. For example, the sensing capability of orthogonal frequency division multiplexing (OFDM) has been examined in the literature, where cyclic prefix (CP) OFDM carrying random data payloads was directly used to sense targets, and a corresponding symbol-wise division-based sensing method was proposed \cite{sturm2011waveform}. The OFDM-ISAC waveform exhibits superior communication and sensing performance in static or quasi-static scenarios, but its performance deteriorates in high-mobility scenarios.


To improve both sensing and communication performances in high-mobility scenarios, a variety of emerging communication-centric signaling strategies have been considered as ISAC waveforms, e.g., orthogonal chirp-division multiplexing (OCDM) \cite{Ouyang2016Orthogonal, Oliveira2020An} and orthogonal time-frequency space (OTFS) \cite{Raviteja2018Interference, Gaudio2020On}. Specifically, the OCDM waveform multiplexes a set of orthogonal chirps, which are complex exponentials whose instantaneous frequencies vary linearly, based on the discrete Fresnel transform (FrT) \cite{Ouyang2016Orthogonal}. For sensing performance, the results in \cite{Oliveira2020An} showed that the sidelobe level of the radar image of the OCDM-ISAC waveform slightly increases compared with the OFDM-ISAC waveform when carrying random communication data. However, OCDM can only achieve partial communication diversity in doubly selective channels \cite{bemani2023affine}. In contrast, OTFS has been shown to achieve full communication diversity in doubly selective channels \cite{Raviteja2018Interference, Gaudio2020On}. In addition to that, the sensing performance of OTFS-ISAC waveforms has been investigated, which exhibits slight degradation in range and velocity estimation accuracy compared with OFDM and frequency modulated continuous wave (FMCW) waveforms \cite{Gaudio2020On}.


More recently, a new affine frequency division multiplexing (AFDM) waveform has been proposed for high-mobility communications. This method multiplexes communication data in the discrete affine Fourier transform (DAFT) domain and can achieve full diversity in doubly selective channels \cite{bemani2023affine,ni2025Integrated,bemani2024integrated,yin2025joint}. Moreover, the results in \cite{bemani2023affine} showed that, compared with OTFS, AFDM achieves comparable BER performance while offering lower complexity and higher spectral efficiency. Towards that end, extensive research has been conducted to further explore the potential of AFDM \cite{yin2024diagonally,luo2024afdm,tao2025affine}. For instance, pilot-aided channel estimation using guard intervals (GI) and equalization algorithms in the DAFT domain have been proposed for AFDM \cite{bemani2023affine}, which has then been extended from the single-antenna scenario to the multiple-input multiple-output (MIMO) scenario \cite{yin2024diagonally}. Moreover, an AFDM-empowered sparse code multiple access (SCMA) system, referred to as AFDM-SCMA, was studied to boost the spectral efficiency in massive connectivity scenarios \cite{luo2024afdm}. The results in \cite{luo2024afdm} showed that the proposed AFDM-SCMA significantly outperforms OFDM-SCMA in uncoded and coded systems. Due to inheriting the advantageous radar properties of chirp signals, AFDM is anticipated to exhibit superior communication and sensing performances simultaneously, especially in high-mobility scenarios.


Due to these compelling advantages, AFDM is regarded as a promising ISAC waveform. The sensing and communication performance bounds and trade-offs of the AFDM-ISAC waveform have been analyzed in \cite{ni2025Integrated}, in terms of communication and sensing spectral efficiency as well as the maximum tolerable delay and Doppler shift. Moreover, an efficient sensing parameter estimation method in the DAFT domain has been proposed, which utilizes all symbols (including pilots and data) of the AFDM-ISAC waveform, demonstrating excellent sensing performance even under large Doppler scenarios \cite{ni2025Integrated}. Additionally, an AFDM-based ISAC scheme that relies on only a single DAFT-domain pilot symbol for sensing has been studied, which exhibits slight performance degradation compared to the scheme using all symbols \cite{bemani2024integrated}.

The ambiguity function (AF) is well-recognized as a key performance indicator of a waveform. For classical predetermined sensing waveforms, e.g., LFM and FMCW, the AFs are deterministic and have closed-form expressions. However, compared with classical sensing waveforms, ISAC waveforms exhibit the following characteristics in practical systems: \romannumeral 1) embedding random communication data, and \romannumeral 2) applying a pulse shaping (PS) filter to restrict the bandwidth of the signal and to eliminate interference among communication symbols \cite{lu2025sensing,liu2025uncovering}. These characteristics may substantially affect the sensing performance. For instance, the randomness of communication data introduces significant variability in the sidelobe levels of the AF, while PS may distort the shape of both the mainlobe and sidelobes of the AF. Therefore, it is crucial to analyze the statistical properties of the AF of pulse-shaped random ISAC waveforms.

Recent literature has investigated the AFs of ISAC waveforms with random communication data and PS. Since sensing signal processing is generally implemented in the discrete time domain, discrete AFs have been widely studied. The authors in \cite{liu2024ofdm} first derived the closed-form expression of the average squared delay cut of the discrete periodic/aperiodic ambiguity function, i.e., the discrete periodic/aperiodic auto-correlation functions (PACF/AACF), of the random OFDM-ISAC waveform. Results in \cite{liu2024ofdm} showed that the AF of a random ISAC waveform is affected by the fourth-order moment, i.e., kurtosis, of the random communication data, and the OFDM-ISAC waveform can achieve the lowest delay-cut ranging sidelobe. The delay-Doppler integrated sidelobe level (ISL) of the discrete periodic ambiguity function (DPAF) of the MIMO-OFDM ISAC waveform is minimized by the symbol-level precoding method \cite{li2024mimo}. Based on this, the authors in \cite{liao2025pulse} investigated the AF of the pulse-shaped random SC-ISAC waveform and proposed a randomness-aware PS method to minimize the average sidelobe level of the AF. An in-depth and systematic analysis of the ranging performance of pulse-shaped random ISAC signals under arbitrary modulation schemes was presented \cite{liu2025uncovering}. The derived closed-form expression of the PACF under a Nyquist PS filter revealed that the average squared PACF can be metaphorically represented as an ``iceberg in the sea'', where the ``iceberg'' corresponds to the squared PACF of the chosen PS response, and the ``sea level'' reflects the variability of the random data.

More relevant to this work, recent studies have shown growing interest in the AF of the AFDM-ISAC waveform. In \cite{zhu2023low}, the AF of a special AFDM waveform was investigated, where the entire communication data was deterministically set to one. In parallel with our earlier conference version \cite{ni2025AFDM}, some independent research on the AF of the random AFDM-ISAC waveform has emerged. The authors in \cite{zhang2025afdm} formulated the DPAF of the pilot-assisted AFDM waveform and derived two conditions that minimize the variance of the AF of AFDM. The continuous form of the AF of a random AFDM-ISAC waveform was studied, and the amplitude of the derived AF was approximated by the Rice distribution \cite{bedeer2025ambiguity}. Moreover, the authors in \cite{yin2025ambiguity} investigated the AFs of continuous-time AFDM signals, revealing the spike-like local property and periodic-like global property of the AF of AFDM pilot subcarriers, as well as the thumbtack-like characteristic of the AF of AFDM symbols.

While the existing works on the AF of AFDM waveforms have provided valuable insights, they generally overlooked the impact of practical PS filters and lacked a detailed examination of sidelobe characteristics, resulting in an incomplete analysis of the sensing performance of AFDM. Moreover, previous works in \cite{liu2025uncovering} and \cite{liao2025pulse} only analyzed the impact of random communication data and PS on the ACF, thus capturing only part of the complete AF. Therefore, it remains necessary to investigate the closed-form expression of the whole AF of the pulse-shaped random AFDM-ISAC waveform.

This paper comprehensively investigates the AF of the pulse-shaped random AFDM-ISAC waveform. The closed-form expressions for the average squared DPAFs of random AFDM, OFDM, and OCDM waveforms with and without PS are respectively derived. On this basis, the impact of the AFDM parameters and the PS filter on the DPAF is analytically characterized, which motivates a novel design methodology for the AFDM parameters that can significantly enhance sensing performance. Numerical results are in complete agreement with the theoretical analysis and verify that AFDM achieves superior velocity estimation performance compared to OFDM in the strong-weak target scenario. For clarity, we summarize our contributions as follows:

\begin{itemize}
	{
		\item We derive a closed-form expression for the average squared DPAF of random AFDM-ISAC signals without PS. The impact of the AFDM parameters on the AF is analyzed, revealing that the DPAF of AFDM is influenced by both the parameter $c_1$ and the kurtosis of random communication data, but is independent of the parameter $c_2$. 
		
		\item On this basis, we comprehensively compare the DPAFs of the AFDM waveform with conventional OFDM and OCDM waveforms. We reveal that there exist the same number of regular depressions in the sidelobes of the DPAFs of AFDM, OCDM, and OFDM waveforms, which may severely deteriorate the detection performance of the weak target in the presence of the strong target. To address this issue, we demonstrate that the AFDM waveform can flexibly control the positions of these depressions by appropriately adjusting the parameter $c_1$, thereby mitigating the negative impact of the depressions caused by the strong target on the detection and estimation of the weak target. 
		
		\item We derive a closed-form expression for the average squared DPAF of pulse-shaped random AFDM-ISAC signals. Theoretical analysis reveals that the PS filter results in a shaped mainlobe along the delay axis determined by the squared PACF of the PS response, and rapidly decaying sidelobes along the Doppler axis determined by the squared spectrum of the squared envelope (SSE) of the PS response.
	}	
\end{itemize} 

The rest of this paper is organized as follows. Section II introduces the preliminaries. In Section III, we derive the AFs of random AFDM, OFDM, and OCDM waveforms without PS and propose a design guideline for the AFDM parameter. Section IV presents the AF of the pulse-shaped random AFDM waveform. Numerical results are presented in Section V. Section VI concludes the paper.

Notation: Throughout the paper, $\mathbf{X}$, $\mathbf{x}$, and $x$ denote a matrix, vector, and scalar, respectively. ${\left\langle \cdot \right\rangle _ N}$, $\delta\left(\cdot\right)$, $\odot$, $\left( \cdot \right)^{*}$, $\left( \cdot \right)^{ T}$, and $\left( \cdot \right)^{ H}$ are the modulo $N$ operation, the Dirac delta function, the Hadamard product, the conjugate operation, the transpose operation, and the Hermitian transpose operation, respectively. $\mathrm{diag}\left(\mathbf{x}\right)$ returns a diagonal matrix with the elements of $\mathbf{x}$ on the main diagonal. $\mathrm{vec}\left(\mathbf{X}\right)$ denotes the vectorization of $\mathbf{X}$. The $n$-th entry of a vector $\mathbf{x}$ and the $\left(m,n\right)$-th entry of a matrix $\mathbf{X}$ are denoted as $x_n$ and $x_{m,n}$, respectively. The Dirichlet function with period $N$ is defined as
\begin{equation*}
\mathcal{{D}}_{N}\left(x\right) = \frac{{\sin \left( {\pi x} \right)}}{{\sin \left( {{{\pi x} \mathord{\left/
					{\vphantom {{\pi x} N}} \right.
					\kern-\nulldelimiterspace} N}} \right)}}.
\end{equation*}

\section{Preliminaries}

\subsection{AFDM with Random Communication Data}

Define $\mathbf{s} \in \mathbb{C}^{N \times 1}$ as the transmitted random data vector, whose entries are independently and identically distributed (i.i.d.) according to a pre-defined complex constellation $\mathcal{S}$, e.g., quadrature amplitude modulation (QAM) or phase shift keying (PSK).
According to \cite{liu2024ofdm}, we impose the following generic assumptions on the adopted constellation.


\begin{assumption} \label{asp:unit_power}
	(Unit Power) We focus on constellations with unit power, namely,\cite{liu2024ofdm}
	\begin{equation} \label{eq:unit_power}
	{\mathbb{E}}\left( {{{\left| {{{s} }} \right|}^2}} \right) = 1, \ \forall s \in \mathcal{{S}}.
	\end{equation}
\end{assumption}
\begin{assumption} \label{asp:rot_sym}
	(Rotational Symmetry) The expectation and pseudo-variance of the constellation are zero, namely,\cite{liu2024ofdm}
	\begin{equation} \label{eq:rot_sym}
	{\mathbb{E}}\left( {{{s}}} \right) = 0,\quad{\mathbb{E}}\left( {{{s}}^2} \right) = 0,\ \forall s \in \mathcal{{S}}.
	\end{equation}
\end{assumption}
If the constellation has zero mean and unit power, its \emph{kurtosis}, i.e., the 4th-order moment, is  
\begin{equation}
\frac{{\mathbb{E}\left\{ {{{\left| {{{s}} - {\mathbb{E}}\left( {{{s}}} \right)} \right|}^4}} \right\}}}{{\mathbb{E}{{\left\{ {{{\left| {{{s}} - {\mathbb{E}}\left( {{{s}}} \right)} \right|}^2}} \right\}}^2}}} = \mathbb{E}\left\{ {{{\left| {{{s}}} \right|}^4}} \right\} \buildrel \Delta \over = {\mu _4} ,\ \forall s \in \mathcal{{S}}.
\end{equation} \normalsize
The kurtosis is equal to 1 for all PSK constellations and is between 1 and 2 for all QAM constellations, which is referred to as a sub-Gaussian constellation\cite{liu2024ofdm}. This paper mainly focuses on sub-Gaussian constellations, i.e., ${\mu _4} < 2$. 
\begin{lemma} \label{Lemma:aver_symbol}
	For i.i.d. constellation data that meets Assumptions 1 and 2, and $m,m',n,n' \in \left[0,N-1\right]$, a useful result was derived in \cite{liu2024ofdm}, i.e.,
	\begin{align}\label{eq:aver_symbol_s}
	&\mathbb{E}\left\{ {{{s}_m^{*}} s_{m'}^{} s_n^{} {{s}_{n'}^*}} \right\} \nonumber \\
	& = \delta \left(m-m'\right) \delta \left(n-n'\right)   + \delta \left(m-n\right)  \delta \left(m'-n'\right)  \nonumber \\
	&  \quad + \left( {{\mu _4} - 2} \right)\delta \left(m-m'\right) \delta  \left(m-n\right) \delta \left(m-n'\right) .
	\end{align}
\end{lemma}

Next, we briefly review the basic principles of the AFDM. For a given two-dimensional (2D) random communication data block $\mathbf{S} = \left[\mathbf{s}_0,\cdots,\mathbf{s}_{N_{\rm sym}-1}\right] \in \mathbb{C}^{N \times N_{\rm sym}}$, the $N$-point inverse discrete affine Fourier transform (IDAFT) is then applied to each column of $\mathbf{S}$ to generate $N_{\rm sym}$ AFDM symbols in the time domain\cite{bemani2023affine,ni2025Integrated}
\begin{equation}\label{eq:AFDM}
{x}_{n,k} = \frac{1}{{\sqrt N }}\sum\nolimits_{m = 0}^{N - 1} {s_{m,k}} {\phi _n}\left( m \right)  ,
\end{equation} 
where ${\phi _n}\left( m \right) = {e^{j2\pi \left( {{c_1}{n^2} + \frac{mn}{N} + {c_2}{m^2}} \right)}}$ with $c_1$ and $c_2$ being the AFDM parameters, $n\in \left[0,N-1\right]$ and $k\in \left[0,N_{\rm sym}-1\right]$. This paper assumes that $c_1$ is chosen such that $2Nc_1$ is an integer\cite{bemani2023affine}.

Then, to ensure the cyclic shift property of the received sensing echo and communication signal, a chirp-periodic prefix (CPP) of length $N_{\rm cp}$ is appended to each of the $N_{\rm sym}$ AFDM symbols, defined as\cite{bemani2023affine}
\begin{equation}\label{eq:symbol_AFDM}
\mathbf{X}_{\rm cp}\left[m,k\right] {=} {e^{ - i2\pi {c_1}\left( {{N^2} + 2Nn} \right)}} {x}_{{{\left\langle n \right\rangle }_N},k},n \in \left[- {N_{\rm cp}},-1\right] .
\end{equation}
It is assumed that the maximum delay introduced by the sensing target or the communication channel is less than the duration of the CPP \cite{bemani2023affine,ni2025Integrated}. After parallel-to-serial conversion, the discrete-time AFDM signal vector $\mathbf{\tilde x} \in \mathbb{C}^{\left(N+N_{\rm cp}\right)N_{\rm sym} \times 1}$ is given by
\begin{equation}\label{eq:symbol_CP_AFDM}
\mathbf{\tilde x} = {\rm vec}\left({\left[ {{\mathbf{X}^{ T}_{\rm cp}}}, \mathbf{X}^{ T}
	\right]^{ T}}\right).
\end{equation}

\subsection{AFDM with Pulse Shaping}

To generate the baseband signal, the discrete-time signal is further passed through a PS filter $\tilde g(t)$, yielding the continuous-time signal $\tilde x_{\rm ps}(t)$ for transmission over the band-limited ISAC channel. Typically, band-limited prototype Nyquist pulses with a one-sided bandwidth $B$ and a roll-off factor $\alpha$ are employed as the PS filter. The resulting continuous-time signal $\tilde x_{\rm ps}(t)$ can be expressed as \cite{liu2025uncovering}
\begin{eqnarray}
\tilde x_{\rm ps}\left( t \right) = \sum\limits_{n = 0}^{N_{\rm all} - 1} {\tilde x_n \tilde g\left( {t - nT} \right)},
\end{eqnarray} 
where $N_{\rm all} = \left(N+N_{\rm cp}\right)N_{\rm sym}$, and $T = \frac{1+\alpha}{2B}$. By employing the unit impulse function $\delta(t)$ in the continuous-time domain, the signal $\tilde x_{\rm ps}(t)$ can be equivalently represented in a linear convolution form as follows:\cite{liu2025uncovering}
\begin{eqnarray}
\tilde x_{\rm ps}\left( t \right) = \sum\limits_{n = 0}^{N_{\rm all} - 1} {\tilde x_n\delta \left( {t - nT} \right) * \tilde g\left( t \right)} .
\end{eqnarray}		
Given the fact that baseband signals are usually processed in the discrete time domain, following \cite{liu2025uncovering}, we proceed with our study using an oversampling-based implementation in the discrete time domain with a sampling rate $f_s = \frac{1}{T_s}$ and a sampling duration $T_s$. Accordingly, $L = \frac{T}{T_s}$ denotes the over-sampling ratio and is assumed an integer. As such, the $k$-th sample of pulse-shaped signal $\tilde x_{\rm ps}\left( t \right)$ is written as\cite{liu2025uncovering}
\begin{eqnarray}\label{eq:pulse_shaping_0}
{{\tilde x_{{\rm ps},k}}} = \tilde x_{\rm ps}\left( {k{T_s}} \right) = \sum\limits_{n = 0}^{N_{\rm all} - 1} {\tilde x_n\delta \left( {k{T_s} - nT} \right) * {\tilde g _k}} ,
\end{eqnarray}
where ${\tilde g _k} = \tilde g\left( {k{T_s}} \right)$.

\begin{figure}[!htbp]
	\centering
	\includegraphics[width=2.0in]{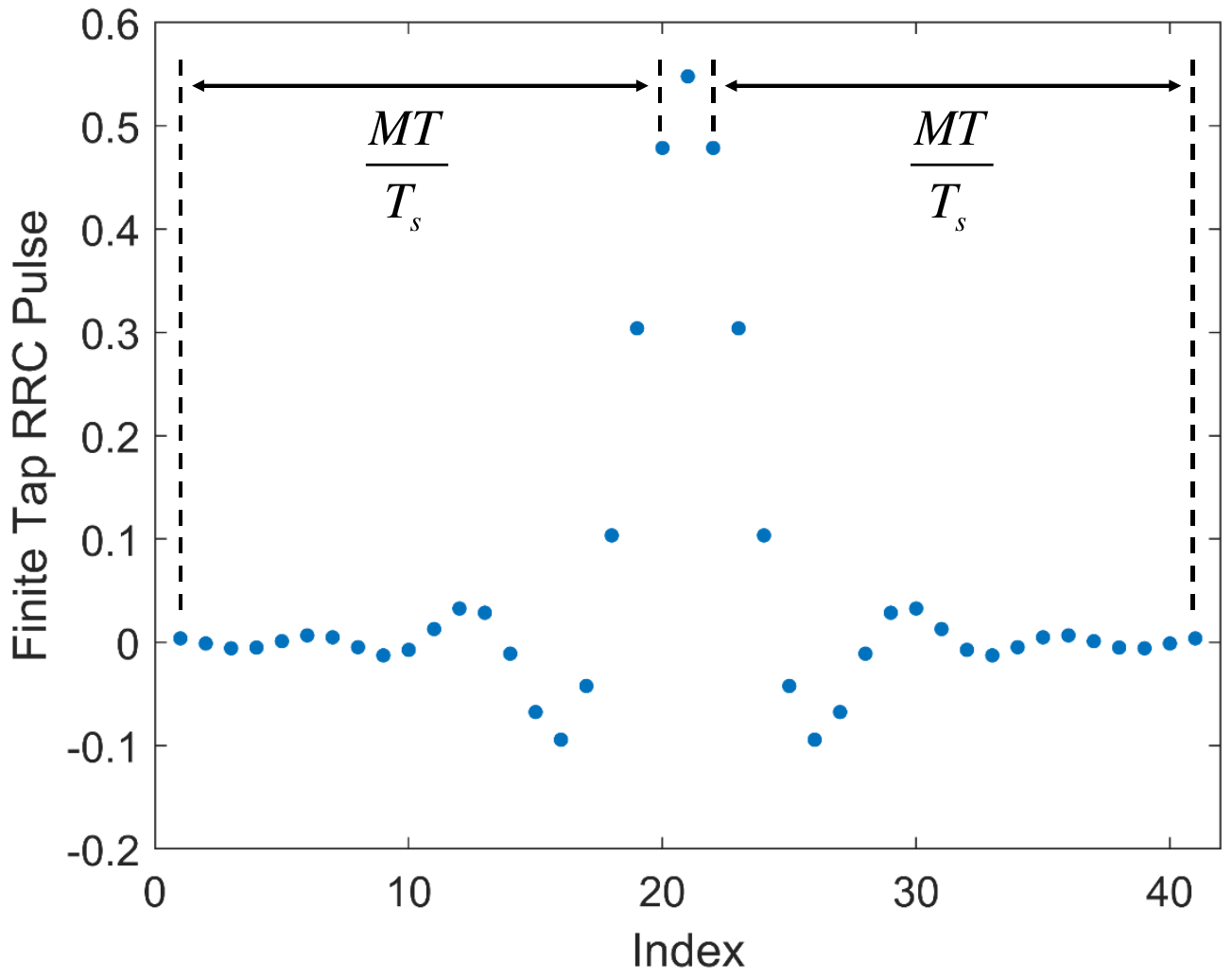}	
	\caption{A finite-tap RRC pulse with $M=5$ and $L=4$.  
		\label{fg:FiniteRap_RRC_N_10_L_4}}
	\vspace*{-5pt} 
\end{figure}

In practical implementations, the amplitude response of the PS filter decays rapidly, and hence the PS filter used in practice typically has a finite number of taps \cite{lin2022orthogonal}. 
This paper considers the PS filter with $2ML+1$ taps, i.e., $\mathbf{\tilde g} = {[{\tilde g _0},{\tilde g _1}, \cdots ,{\tilde g _{\left({2ML}\right)}}]^{{T}}} \in \mathbb{R}^{\left(2ML+1\right)\times 1}$ with ${\left\| \mathbf{\tilde g} \right\|^2} = 1$. An example of a finite-tap RRC PS filter is shown in Fig. \ref{fg:FiniteRap_RRC_N_10_L_4}. 

\begin{figure}[!htbp]
	\centering
	\includegraphics[width=2.6in]{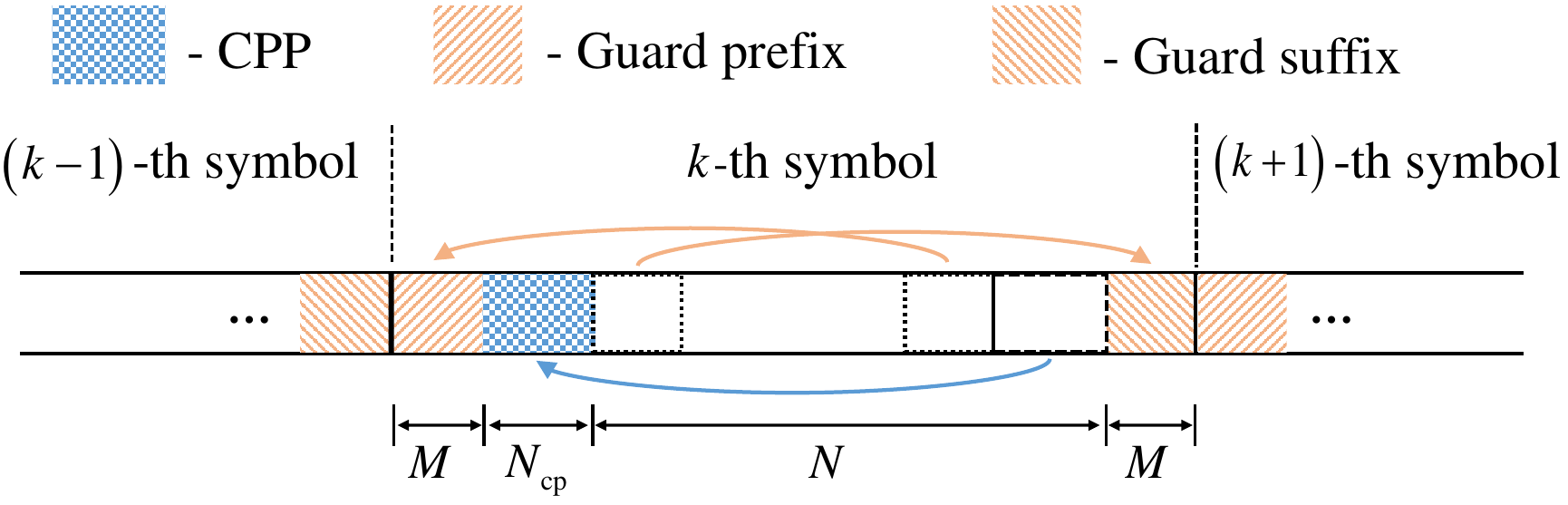}	
	\caption{The framework of GPS-AFDM.  
		\label{fg:GPS_CP_AFDM}}
	\vspace*{-5pt} 
\end{figure}
However, if the discrete-time AFDM signal in (\ref{eq:symbol_CP_AFDM}) is directly filtered by the PS filter $\mathbf{\tilde g}$, the received signal loses its cyclic shift property after propagating over the ISAC channel, thereby introducing inter-symbol interference. To address this issue, we insert the guard prefix and suffix (GPS) segments into the AFDM signal, referred to as GPS-AFDM, to eliminate the interference caused by PS. Specifically, the framework of the GPS-AFDM signal is shown in Fig. \ref{fg:GPS_CP_AFDM}. 
The guard prefix and guard suffix of the $k$-th AFDM symbol are given by
\begin{align}
&{x_{{\rm{gp}},n,k}} = {e^{ - i2\pi {c_1}\left( {{N^2} + 2Nn} \right)}} x_{ {{\left\langle n \right\rangle }_N},k },\\
&{x_{{\rm{gs}},m,k}} = {e^{ - i2\pi {c_1}\left( {{N^2} - 2Nm} \right)}} x_{ {{\left\langle m \right\rangle }_N},k },
\end{align}
where $n =  - {N_{\rm cp}} - M, \ldots , - {N_{\rm cp}} - 1$, and $m = N + 1, \ldots ,N + M$. After Parallel-to-Serial conversion, let $N_{\rm gps} = \left(N+N_{\rm cp}+2M\right)N_{\rm sym}$, the discrete-time GPS-AFDM signal vector $\mathbf{\tilde x}_{\rm gps} \in \mathbb{C}^{N_{\rm gps} \times 1}$ is given by
\begin{equation}\label{eq:symbol_GPS_CP_AFDM}
\mathbf{\tilde x}_{\rm gps}  = {\rm vec}\left({\left[ 
	{{\mathbf{X}^{ T}_{\rm gp}}}, {{\mathbf{X}^{ T}_{\rm cp}}},\mathbf{X}^{ T},{{\mathbf{X}^{ T}_{\rm gs}}}
	\right]^{ T}}\right).
\end{equation}

Then, according to (\ref{eq:pulse_shaping_0}), by defining the up-sampled GPS-AFDM signal as
\begin{equation}
{\mathbf{\tilde x}_{\rm gps, up}} = {\left[ {{\tilde x_{{\rm gps},0}},\mathbf{0}_{L - 1}^{{T}}, \cdots ,{\tilde x_{{\rm gps},N_{\rm gps}-1}},\mathbf{0}_{L - 1}^{{T}}} \right]^{{T}}},
\end{equation}
the pulse-shaped signal can be expressed in matrix form as:
\begin{equation}\label{eq:pulse_shaping}
{\mathbf{\tilde x}_{\rm ps}} = \mathbf{\tilde G}{\mathbf{\tilde x}_{\rm gps, up}},
\end{equation}
where $\mathbf{\tilde x}_{\rm ps} \in \mathbb{C}^{\left(N_{\rm gps}+2M\right)L\times 1}$ denotes the pulse-shaped AFDM signal, and ${\mathbf{\tilde G}} \in {\mathbb{R}^{\left(N_{\rm gps}+2M\right)L \times  N_{\rm gps} L}}$ denotes the aperiodic PS matrix, which is given by
\begin{equation}\label{eq:pulse_shaping_matrix}
{\mathbf{\tilde G}} = \left[ {\begin{array}{*{20}{c}}
	{\tilde g _ 0}&0& \cdots &0\\
	\vdots &{\tilde g _ 0}&{}& \vdots \\
	{\tilde g_ {NL - 1} }& \vdots & \ddots &0\\
	0&{\tilde g_ {NL - 1}}&{}&{\tilde g _ 0}\\
	\vdots & \vdots &{}& \vdots \\
	0&0& \cdots &{\tilde g_ {NL - 1}}
	\end{array}} \right]  .
\end{equation} 

\subsection{Received Echo of Pulse-Shaped AFDM}

Considering a sensing channel with $Q$ point-like targets, the received echo vector $\mathbf{\tilde y} \in \mathbb{C}^{\left(N_{\rm gps}+2M\right)L\times 1}$ can be expressed as
\begin{equation}
\mathbf{\tilde y} = \sum\nolimits_{q = 1}^Q {{{\mathbf{\tilde B}} _q}{\mathbf{{\tilde \Delta }}_{{\nu _q}}}{\mathbf{{\tilde J}}_{{\tau _q}}}{\mathbf{{\tilde x}}_{\rm ps}}}  + \mathbf{\tilde z},
\end{equation}
where $\mathbf{\tilde z}$ denotes the noise in the time domain, and ${\tau _q}=\frac{1}{LT}{\tilde \tau _q}$ and ${\nu _q}=\frac{{1}}{{{\Delta  _f}}}\tilde \nu _q$ are the normalized delay and normalized Doppler frequency shift of the $q$-th target, respectively. Here ${\tilde \tau _q}$ and ${\tilde \nu _q}$ denote the continuous delay and Doppler frequency shift of the $q$-th target, and $\Delta _f$ is the subcarrier spacing. Following \cite{liu2025uncovering}, while this paper mainly focuses on the case of integer normalized delay, it yields a practical and accurate approximation of the actual problem, especially when the over-sampling ratio $L$ is large. The target response matrix ${{\mathbf{\tilde B}} _q} \in \mathbb{C}^{\left( {{N_{\rm gps}} + 2M} \right)L \times \left( {{N_{\rm gps}} + 2M} \right)L}$ of the $q$-th target can be written as
\begin{equation}
{{\mathbf{\tilde B}} _q} = \mathrm{diag}\left(\tilde \beta _{q,n},\ n \in \left[ 0,\left(N_{\rm gps}+2M\right)L-1 \right] \right),
\end{equation}
where $\tilde \beta _{q,n}$ denotes the complex reflection coefficient of the $q$-th target at the $n$-th snapshot. For the nonfluctuating (Swerling 0) target, ${{\mathbf{\tilde B}} _q} = {\bar \beta _{q}} \mathbf{I}$ with ${\bar \beta _{q}}$ being the average complex reflection coefficient of the $q$-th target. However, due to the high velocity and maneuverability of targets in high-mobility scenarios, targets may exhibit fast fluctuations, e.g. following Swerling 2 model, a phenomenon known as pulse-to-pulse decorrelation, where each pulse corresponds to one AFDM symbol\cite[Sec. 7.5]{mark2010principles}. As such, following \cite{mark2010principles}, this paper assumes that the target reflection coefficient remains constant within each received pulse-shaped AFDM symbol and changes independently from pulse to pulse. The Doppler matrix ${\mathbf{{\tilde \Delta }}_{{\nu _q}}}$ is defined as
\begin{equation}
{\mathbf{{\tilde \Delta }}_{{\nu _q}}} = \mathrm{diag}\left( {{e^{ j\frac{2\pi}{NL} {\nu _q}n}}},n \in \left[ 0,\left(N_{\rm gps}+2M\right)L-1 \right] \right),
\end{equation}
and ${\mathbf{{\tilde J}}_{{\tau _q}}} \in {\mathbb{C}^{\left( {{N_{\rm gps}} + 2M} \right)L \times \left( {{N_{\rm gps}} + 2M} \right)L}}$ denotes the aperiodic time-shift matrix, which is given by\cite{liu2025uncovering}
\begin{equation}
{\mathbf{{\tilde J}}_{{\tau _q}}} = \left[ {\begin{array}{*{20}{c}}
	\mathbf{0}&\mathbf{0}\\
	{{\mathbf{I}_{\left( {{N_{\rm gps}} + 2M} \right)L - {\tau _q}}}}&\mathbf{0}
	\end{array}} \right].
\end{equation}

Before the sensing signal processing, we only retain the valid signal with $N_{\rm gps}L$ samples in the sampling range $n\in\left[ML,N_{\rm gps}L+ML-1\right]$ by removing the transition segment introduced by the PS filter. After that, performing Serial-to-Parallel conversion and then removing the guard prefix, CPP and guard suffix of each symbol, the received echo signal matrix $\mathbf{Y}\in {\mathbb{C}^{NL\times N_{\rm sym}}}$ can be written as
\begin{equation} \label{eq:rec_Matrix}
\mathbf{Y} = \sum\limits_{q = 1}^Q {{e^{ j\frac{2\pi}{NL} {\nu _q}\left(N_{\rm cp}+2M\right)L}}{\mathbf{{\Delta }}_{{\nu _q}}}{\mathbf{{J}}_{{\tau _q}}}{\mathbf{G}}{\mathbf{ X}_{\rm up}}}{{\mathbf{ B}} _q}{\mathbf{{D}}_{{\nu _q}}}  + \mathbf{Z},
\end{equation}
where $\mathbf{Y} = \left[\mathbf{y}_0,\cdots,\mathbf{y}_{N_{\rm sym}-1}\right]$, and $\mathbf{Z}$ denotes the noise matrix. ${\mathbf{{\Delta }}_{{\nu _q}}} \in \mathbb{C}^{NL\times NL}$ and $\mathbf{D}_{\nu _q} \in \mathbb{C}^{N_{\rm sym}\times N_{\rm sym}}$ represent the fast-time Doppler matrix and the slow-time Doppler matrix, respectively, defined as
\begin{align}
&{\mathbf{{\Delta }}_{{\nu _q}}} = \mathrm{diag}\left( {{e^{ j\frac{2\pi}{NL} {\nu _q}n}}} \right), \\
&\mathbf{D}_{\nu _q}= \mathrm{diag}\left( {{e^{ j\frac{2\pi}{NL} {\nu _q}\left(N+N_{\rm cp}+2M\right)Lk}}} \right),
\end{align} 
where $n \in \left[ 0,NL-1 \right]$ and $k \in \left[ 0,N_{\rm sym}-1 \right]$.
${{\mathbf{B}} _q} \in \mathbb{C}^{N_{\rm sym} \times N_{\rm sym}}$ represents the target response matrix, i.e.,
\begin{equation}
{{\mathbf{ B}} _q} = \mathrm{diag}\left( \beta _{q,k},\ k \in \left[ 0,N_{\rm sym}-1 \right] \right),
\end{equation}
where $\beta _{q,k}$ denotes the complex reflection coefficient of the $q$-th fast fluctuating Swerling 2 target for the $k$-th pulse-shaped AFDM symbol, which follows the complex Gausion distribution, i.e., $\beta _{q,k} \sim \mathcal{CN}\left(0, \bar \beta _{q}^2\right)$\cite{mark2010principles}.
${\mathbf{{J}}_{{\tau _q}}}\in \mathbb{C}^{NL\times NL}$ deontes the periodic time-shift matrix, which is\cite{liu2025uncovering}
\begin{equation}
{\mathbf{{J}}_{{\tau _q}}} = \left[ {\begin{array}{*{20}{c}}
	\mathbf{0}&\mathbf{I}_{\tau _q}\\
	{{\mathbf{I}_{NL - {\tau _q}}}}&\mathbf{0}
	\end{array}} \right],
\end{equation} 
and ${\mathbf{G}} \in {\mathbb{R}^{NL \times NL}}$ is the periodic PS matrix, defined as	
\begin{equation}\label{eq:pulse_shaping_matrix_circ}
{\mathbf{G}} = \left[ {\begin{array}{*{20}{c}}
	{\tilde g_{ML}}&{\tilde g_{ML - 1}}& \cdots &{\tilde g_{ML + 1}}\\
	\vdots & {\tilde g_{ML}} &{}& \vdots \\
	{\tilde g_{2ML} }&{\vdots}&{}&{ }\\
	0&{\tilde g_{2ML} }&{}&0\\
	\vdots &0& \ddots & \vdots \\
	0& \vdots &{}&0\\
	{\tilde g _ 0 }&0&{}&{}\\
	\vdots & \vdots &{}& \vdots \\
	{\tilde g_{ML - 1} }&{\tilde g_{ML - 2}}& \cdots &{\tilde g_{ML}}
	\end{array}} \right] .
\end{equation} 
Moreover, $\mathbf{X}_{\rm up}$ is the up-sampled signal matrix of ${\mathbf{ X}}$, expressed as
\begin{equation}
{\mathbf{X}_{\rm up}} = \left[ {{\mathbf{x}_{{\rm up},0}}, \cdots ,{\mathbf{x}_{{\rm up},{N_{\rm sym}} - 1}}} \right],
\end{equation}
where ${\mathbf{x}_{{\rm up},k}} = {\left[ {{x_{0,k}},\mathbf{0}_{L - 1}^{{T}}, \cdots ,{x_{N-1,k}},\mathbf{0}_{L - 1}^{{T}}} \right]^{{T}}}$. Let ${\mathbf{x}_{{\rm ps},k}}$ denote the $k$-th pulse-shaped AFDM symbol, which is\cite{liu2025uncovering}
\begin{equation} \label{eq:pulsed_AFDM}
{\mathbf{x}_{{\rm ps},k}} = {\mathbf{G}} {\mathbf{x}_{{\rm up},k}}.
\end{equation}
Then, the $k$-th received AFDM symbol in (\ref{eq:rec_Matrix}) is given by
\begin{align}\label{eq:received_signal}
&{y}_k\left[n\right] = \sum\limits_{q = 1}^Q {\beta _{q,k}}{e^{ j\frac{2\pi}{NL} {\nu _q}\left(N_{\rm cp}+2M\right)L}}{e^{ j\frac{2\pi}{NL} {\nu _q}\left(N+N_{\rm cp}+2M\right)Lk}} \nonumber \\
& \quad\quad\quad\quad \cdot {{x}_{{\rm ps},k}}\left[{{\left\langle n - {\tau _q} \right\rangle }_{NL}}\right]{e^{ j\frac{2\pi}{NL} {\nu _q}n}}  + z_{n,k},
\end{align}
where $k\in \left[0,N_{\rm sym}-1\right]$ and $n\in \left[0,NL-1\right]$. 
It is observed from (\ref{eq:received_signal}) that, under the GPS-AFDM scheme, each received symbol $\mathbf{y}_k$ constitutes a periodic shift counterpart of the pulse-shaped AFDM symbol ${\mathbf{x}_{{\rm ps},k}}$. Hence, for each received symbol $\mathbf{y}_k$, applying a matched filter (MF) through periodic convolution with the reference signal ${\mathbf{x}_{{\rm ps},k}}$ enables the extraction of the delay and Doppler information of sensing targets, yielding
\begin{align}
{{{\tilde r}}_k}\left( {\tau ,\nu } \right) = \sum\limits_{n = 0}^{NL-1} {{y_k}\left[ n \right]x_{{\rm ps},k}^ * \left[ {{{\left\langle {n - \tau } \right\rangle }_{NL}}} \right]{e^{ - j\frac{{2\pi }}{{NL}}\nu n}}}.
\end{align}

Due to the randomness of the transmitted ISAC waveform, the sidelobes of ${{{\tilde r}}_k}\left( {\tau ,\nu } \right)$ exhibit significant fluctuations, which may adversely affect target detection and parameter estimation. Integration of multiple matched filter (MF) outputs can reduce the fluctuations caused by randomness \cite{liu2025uncovering}. However, for fast fluctuating targetsc, the phases of these symbols may be noncoherent. Therefore, this paper primarily considers noncoherent integration. According to \cite{mark2010principles}, the output of noncoherent integration over $N_{\rm sym}$ symbols with a square law detector can be formulated as
\begin{align}\label{eq:noncoh_integ}
	r\left( {\tau ,\nu } \right)  &= \frac{1}{N_{\rm sym}} \sum\nolimits_{k = 0}^{{N_{{\rm{sym}}}} - 1} {{{\left| {{{{{\tilde r}}}_k}\left( {\tau ,\nu } \right)} \right|}^2}}.
\end{align}
When $N_{\rm sym}$ is sufficiently large, $r\left( {\tau ,\nu } \right)$ approaches the expectation of ${{{\left| {{{{{\tilde r}}}_k}\left( {\tau ,\nu } \right)} \right|}^2}}$, which is defined as (\ref{eq:ave_squ_rk}) at the top of next page.
\begin{figure*} [htp]
	\vspace*{-10pt} 
	\small
	\begin{flalign}\label{eq:ave_squ_rk}
	&\ r\left( {\tau ,\nu } \right) \approx	\mathbb{E}\left\{ {{{\left| {{{{{\tilde r}}}_k}\left( {\tau ,\nu } \right)} \right|}^2}} \right\} & \nonumber \\
	&\	\mathop  = \limits^{\left( a \right)} \mathbb{E}\left\{ {{{\left| {\sum\limits_{q = 1}^Q {{\beta _{q,k}}{e^{j\frac{{2\pi }}{{NL}}{\nu _q}\left( {{N_{{\rm{cp}}}} + 2M} \right)L}}{e^{j\frac{{2\pi }}{{NL}}{\nu _q}\left( {N + {N_{{\rm{cp}}}} + 2M} \right)Lk}}\sum\limits_{n = 0}^{NL - 1} {{x_{{\rm{ps}},k}}\left[ {{{\left\langle {n - {\tau _q}} \right\rangle }_{NL}}} \right]x_{\rm{ps},k}^ * \left[ {{{\left\langle {n - \tau } \right\rangle }_{NL}}} \right]{e^{ - j\frac{{2\pi }}{{NL}}\left( {\nu  - {\nu _q}} \right)n}}} } } \right|}^2}} \right\} + \mathbb{E}\left\{ {{{\left| {{{\tilde z}_k}\left( {\tau ,\nu } \right)} \right|}^2}} \right\} & \nonumber \\ 
	&\	\mathop  = \limits^{\left( b \right)} \sum\limits_{q = 1}^Q {\mathbb{E}\left\{ {{{\left| {{\beta _{q,k}}} \right|}^2}} \right\}\mathbb{E}\left\{ {{{\left| {\sum\limits_{n = 0}^{NL - 1} {{x_{{\rm{ps}},k}}\left[ n \right]x_{\rm{ps},k}^ * \left[ {{{\left\langle {n - \left( {\tau  - {\tau _q}} \right)} \right\rangle }_{NL}}} \right]{e^{ - j\frac{{2\pi }}{{NL}}\left( {\nu  - {\nu _q}} \right)n}}} } \right|}^2}} \right\}}  + \mathbb{E}\left\{ {{{\left| {{{\tilde z}_k}\left( {\tau ,\nu } \right)} \right|}^2}} \right\} & \nonumber \\
	&\	\mathop  =  \sum\limits_{q = 1}^Q {\bar \beta _q^2 \cdot \mathbb{E}\left\{ {{{\left| {\chi _k \left( {\tau  - {\tau _q},\nu  - {\nu _q}} \right)} \right|}^2}} \right\}}  + \mathbb{E}\left\{ {{{\left| {{{\tilde z}_k}\left( {\tau ,\nu } \right)} \right|}^2}} \right\}.&
	\end{flalign} \normalsize
		\hrule
	\vspace*{-10pt} 
\end{figure*}
In (\ref{eq:ave_squ_rk}), ${{{\tilde z}_k}\left( {\tau ,\nu } \right)}$ is the noise component, defined as
\begin{equation}
{{{\tilde z}_k}\left( {\tau ,\nu } \right)} = \sum\limits_{n = 0}^{NL-1} {{z}_{n,k} \cdot x_{{\rm ps},k}^ * \left[ {{{\left\langle {n - \tau } \right\rangle }_{NL}}} \right]{e^{ - j\frac{{2\pi }}{{NL}}\nu n}}},
\end{equation}
and $\chi _k \left( {\tau} ,{\nu} \right)$ denotes the DPAF of the pulse-shaped signal vector ${\mathbf{x}_{{\rm ps},k}}$ in (\ref{eq:pulsed_AFDM}), which is given by
\begin{equation}\label{eq:DPAF_k}
\chi _k\left( {\tau} ,{\nu} \right) = \sum\limits_{n = 0}^{NL - 1} {{{ x}_{{\rm ps},k}}\left[ n \right]{{{ x}^ *_{{\rm ps},k}} }\left[ {{{\left\langle {n - \tau} \right\rangle }_{NL}}} \right]{e^{-j\frac{{2\pi }}{NL}{\nu}n}}}. 
\end{equation}
The equality (a) holds because the received zero-mean noise is independent of the signal. Moreover, the equality (b) holds because the zero-mean complex coefficient ${\beta _{q,k}}$ is independent of each other across different $q$ and $k$.
Consequently, the output of noncoherent integration of MF is affected by the average squared DPAF of the pulse-shaped signal, i.e., $\mathbb{E}\left\{ \left|\chi \left( {\tau} ,{\nu} \right)\right|^2  \right\}$. 

\subsection{Impact of randomness and PS on the ACF}

Next, we briefly review the impact of randomness and PS on the ACF according to the Iceberg Theorem in \cite{liu2025uncovering}.
\begin{lemma} \label{lemma:Iceberg}
	(Iceberg Metaphor\cite{liu2025uncovering}) The average squared ACF of random pulse-shaped waveforms consists of two components: the ``iceberg'' and the ``sea level''.
	The ``iceberg'' depends on the PS, which determines the overall shape of the ACF and primarily affects the sensing performance for targets near the mainlobe. In contrast, the ``sea level'' represents the sidelobes exhibiting approximately uniform levels, raised from the randomness of the data payload. Additionally, periodic ripples appear in the sidelobe region of the ACF, which are caused by the PS and can be metaphorically interpreted as ``sea waves'' over the ``sea level''.
\end{lemma}

\section{Closed-Form Expression for Average Squared DPAF of AFDM without PS}

This section investigates the impact of AFDM parameters on the DPAF and compares the DPAFs of different ISAC waveforms. To this end, we first derive a closed-form expression for the average squared DPAF of random AFDM, OFDM, and OCDM signals without PS. Based on this, the impact of AFDM parameters on the DPAF is analyzed. Finally, a DPAF-inspired AFDM parameter design guideline is proposed.

\subsection{Closed-form Expression for DPAFs of AFDM, OFDM and OCDM Waveforms without PS}

\subsubsection{DPAF of AFDM}
For the random AFDM without PS, $L = 1$, $\mathbf{\tilde g} = {[\mathbf{0}_{M}^{{T}},1,\mathbf{0}_{M}^{{T}}]^{{T}}}$, $\mathbf{G}=\mathbf{I}_N$ and ${\mathbf{ x}_{{\rm ps}}} = \mathbf{x}$. As such, the average squared DPAF of AFDM can be rewritten as (\ref{eq:DPAF_AFDM_1}) at the top of next page, where $\bar s _m = s _m{e^{j2\pi {c_2}{m^2}}}$. The closed-form expression for the average squared DPAF of AFDM without PS is shown in the following Proposition.

\begin{figure*}	[htb]
	\vspace*{-10pt} 
	\small
	\begin{flalign} \label{eq:DPAF_AFDM_1}
	&\	\mathbb{E}\left\{ {{{\left| {\chi _{\text{AFDM}} \left( {\tau ,\nu } \right)} \right|}^2}} \right\} = \frac{1}{{{N^2}}}\sum\limits_{n = 0}^{N - 1} {\sum\limits_{n' = 0}^{N - 1} {\sum\limits_{m = 0}^{N - 1} {\sum\limits_{m' = 0}^{N - 1} {\left\{ {\frac{{{e^{ - j2\pi \left( {n' - n - 2N{c_1}\tau  + \nu } \right)}} - 1}}{{{e^{ - j\frac{{2\pi }}{N}\left( {n' - n - 2N{c_1}\tau + \nu } \right)}} - 1}}\frac{{{e^{  j2\pi \left( {m' - m - 2N{c_1}\tau  + \nu } \right)}} - 1}}{{{e^{  j\frac{{2\pi }}{N}\left( {m' - m - 2N{c_1}\tau  + \nu } \right)}} - 1}}} {{e^{ - j\frac{{2\pi }}{N}\left( {m' - n'} \right)\tau}}\mathbb{E}\left\{ {{{\bar s}_m^{*}} \bar s_{m'}^{} \bar s_n^{} {{\bar s}_{n'}^*}} \right\}} \right\} } } } }  .&
	\end{flalign} \normalsize
	\vspace*{-10pt} 
\end{figure*}

\begin{proposition} \label{prop:AS_DPAF_noPS}
	The closed-form expression of the average squared DPAF of AFDM without PS can be expressed as (\ref{eq:DPAF_AFDM_noPulse_fracDop}) at the top of next page.
			\begin{figure*}	[htb]
				\vspace*{-10pt} 
				\small
	\begin{flalign} \label{eq:DPAF_AFDM_noPulse_fracDop}
	&\	\mathbb{E}\left\{ {{{\left| {\chi _{\text{AFDM}} \left( {\tau ,\nu} \right)} \right|}^2}} \right\}  =  \frac{1}{{{N^2}}}{\left| {\frac{{\sin \left[ {\pi \left( {2N{c_1}\tau  - \nu} \right)} \right]}}{{\sin \left[ {{{\pi \left( {2N{c_1}\tau  - \nu} \right)} \mathord{\left/
								{\vphantom {{\pi \left( {2N{c_1}\tau  + \nu} \right)} N}} \right.
								\kern-\nulldelimiterspace} N}} \right]}}} \right|^2}{\left| {\frac{{\sin \left( {\pi \tau} \right)}}{{\sin \left( {{{\pi \tau} \mathord{\left/
								{\vphantom {{\pi \tau} N}} \right.
								\kern-\nulldelimiterspace} N}} \right)}}} \right|^2} 	+ \left( {{\mu _4} - 2} \right)\frac{1}{N}{\left| {\frac{{\sin \left[ {\pi \left( {2N{c_1}\tau  - \nu} \right)} \right]}}{{\sin \left[ {{{\pi \left( {2N{c_1}\tau  - \nu} \right)} \mathord{\left/
								{\vphantom {{\pi \left( {2N{c_1}\tau  - \nu} \right)} N}} \right.
								\kern-\nulldelimiterspace} N}} \right]}}} \right|^2} &\nonumber\\
	&\ \quad\quad\quad\quad\quad\quad\quad\quad\quad\ \ 	+ \frac{1}{N}\sum\limits_{m = 0}^{N - 1} {{{\left| {\frac{{\sin \left[ {\pi \left( {m + 2N{c_1}\tau  - \nu} \right)} \right]}}{{\sin \left[ {{{\pi \left( {m + 2N{c_1}\tau  - \nu} \right)} \mathord{\left/
										{\vphantom {{\pi \left( {m + 2N{c_1}\tau  - \nu} \right)} N}} \right.
										\kern-\nulldelimiterspace} N}} \right]}}} \right|}^2}} .&
	\end{flalign} \normalsize
	\hrule
	\vspace*{-10pt} 
		\end{figure*}
\end{proposition}
\begin{proof}
	See Appendix A.
\end{proof}
When the normalized Doppler $\nu$ is an integer, the average squared DPAF of AFDM reduces to
\begin{flalign} \label{eq:AS_DPAF_AFDM_4}
&\mathbb{E}\left\{ {{{\left| {\chi _{\text{AFDM}} \left( {\tau ,{\nu}} \right)} \right|}^2}} \right\} = {N^2}\delta \left( {{{\left\langle { 2N{c_1}\tau - {\nu}} \right\rangle }_N}} \right)\delta \left( \tau \right) + N \nonumber \\
& \quad\quad\quad\quad\quad\quad\quad\quad\quad + \left( {{\mu _4} - 2} \right)N\delta \left( {{{\left\langle {2N{c_1}\tau - {\nu} } \right\rangle }_N}} \right). 
\end{flalign} 
This result is consistent with Proposition 1 in our earlier conference version \cite{ni2025AFDM}.


\subsubsection{DPAF of OFDM}
Let $c_1 = 0$, AFDM is equivalent to OFDM\cite{bemani2023affine}. As such, the derived average squared DPAF in (\ref{eq:DPAF_AFDM_noPulse_fracDop}) reduces to
\begin{flalign} \label{eq:AS_DPAF_OFDM_1}
&\	\mathbb{E}\left\{ {{{\left| {\chi _{\text{OFDM}} \left( {\tau ,\nu} \right)} \right|}^2}} \right\} =  \frac{1}{{{N^2}}}{\left| {\frac{{\sin \left( {\pi  {   \nu} } \right)}}{{\sin \left( {{{\pi  {\nu} } \mathord{\left/
							{\vphantom {{\pi \left( {2N{c_1}\tau  + \nu} \right)} N}} \right.
							\kern-\nulldelimiterspace} N}} \right)}}} \right|^2}{\left| {\frac{{\sin \left( {\pi \tau} \right)}}{{\sin \left( {{{\pi \tau} \mathord{\left/
							{\vphantom {{\pi \tau} N}} \right.
							\kern-\nulldelimiterspace} N}} \right)}}} \right|^2} &\nonumber\\
&\	+ \left( {{\mu _4} - 2} \right)\frac{1}{N}{\left| {\frac{{\sin \left( {\pi  { \nu} } \right)}}{{\sin \left( {{{\pi  { \nu} } \mathord{\left/
							{\vphantom {{\pi  { \nu} } N}} \right.
							\kern-\nulldelimiterspace} N}} \right)}}} \right|^2}  \hspace{-0.15cm} + \hspace{-0.10cm} \frac{1}{N}\sum\limits_{m = 0}^{N - 1} {{{\left| {\frac{{\sin \left[ {\pi \left( {m -  \nu} \right)} \right]}}{{\sin \left[ {{{\pi \left( {m -  \nu} \right)} \mathord{\left/
									{\vphantom {{\pi \left( {m -  \nu} \right)} N}} \right.
									\kern-\nulldelimiterspace} N}} \right]}}} \right|}^2}} ,&
\end{flalign} 
which is the closed-form expression for the average squared DPAF of OFDM without PS. 

\subsubsection{DPAF of OCDM}
According to\cite{bemani2023affine}, OCDM uses $c_1 = c_2 = {1 \mathord{\left/
		{\vphantom {1 \left(2N\right)}} \right.
		\kern-\nulldelimiterspace} \left(2N\right)}$, i.e., $2Nc_1=1$. Consequently, by applying (\ref{eq:DPAF_AFDM_noPulse_fracDop}), the closed-form expression for the average squared DPAF of OCDM without PS is given by 
\begin{flalign} \label{eq:DPAF_OCDM_noPulse_fracDop}
&\	\mathbb{E}\left\{ {{{\left| {\chi _{\text{OCDM}} \left( {\tau ,\nu} \right)} \right|}^2}} \right\} \hspace{-0.15cm} =  \hspace{-0.15cm} \frac{1}{{{N^2}}}{\left| {\frac{{\sin \left[ {\pi \left( {\tau  - \nu} \right)} \right]}}{{\sin \left[ {{{\pi \left( {\tau  - \nu} \right)} \mathord{\left/
							{\vphantom {{\pi \left( {\tau  - \nu} \right)} N}} \right.
							\kern-\nulldelimiterspace} N}} \right]}}} \right|^2}{\left| {\frac{{\sin \left( {\pi \tau} \right)}}{{\sin \left( {{{\pi \tau} \mathord{\left/
							{\vphantom {{\pi \tau} N}} \right.
							\kern-\nulldelimiterspace} N}} \right)}}} \right|^2} &\nonumber\\
&\	\quad\quad\quad\quad\quad\quad	+ \left( {{\mu _4} - 2} \right)\frac{1}{N}{\left| {\frac{{\sin \left[ {\pi \left( {\tau  - \nu} \right)} \right]}}{{\sin \left[ {{{\pi \left( {\tau  - \nu} \right)} \mathord{\left/
							{\vphantom {{\pi \left( {\tau  - \nu} \right)} N}} \right.
							\kern-\nulldelimiterspace} N}} \right]}}} \right|^2} &\nonumber\\
&\ \quad\quad\quad\quad\quad\quad	+ \frac{1}{N}\sum\limits_{m = 0}^{N - 1} {{{\left| {\frac{{\sin \left[ {\pi \left( {m + \tau  - \nu} \right)} \right]}}{{\sin \left[ {{{\pi \left( {m + \tau  - \nu} \right)} \mathord{\left/
									{\vphantom {{\pi \left( {m + \tau  - \nu} \right)} N}} \right.
									\kern-\nulldelimiterspace} N}} \right]}}} \right|}^2}} .&
\end{flalign} 

\subsection{Analyses on DPAFs of AFDM, OFDM and OCDM} \label{sec:Analy_DPAF}

\subsubsection{Mainlobes of DPAFs of AFDM, OFDM and OCDM}


\begin{corollary} \label{cor:mainlobe_DPAF_AFDM}
	The average squared DPAFs of AFDM, OFDM and OCDM waveforms without PS have the same mainlobe level, which is given by
	\begin{align}\label{eq:mainlobe_DPAF_AFDM}
	\mathbb{E}\left\{ {{{\left| {\chi \left( {0,0} \right)} \right|}^2}} \right\} = {N^2} + \left( {{\mu _4} - 1} \right)N.
	\end{align}	
\end{corollary}
\begin{proof}
	Substituting $\tau=0$, and $\nu=0$ into (\ref{eq:DPAF_AFDM_noPulse_fracDop}), (\ref{eq:AS_DPAF_OFDM_1}) and (\ref{eq:DPAF_OCDM_noPulse_fracDop}) immediately yields (\ref{eq:mainlobe_DPAF_AFDM}), thereby completing the proof.
\end{proof}

\subsubsection{Sidelobes of DPAFs of AFDM, OFDM and OCDM}

\begin{corollary} \label{cor:sidelobe_DPAF_AFDM}
	The average sidelobe level of the squared DPAF of AFDM without PS is expressed as 	
	\begin{align}\label{eq:sidelobe_DPAF_AFDM}
	&\mathbb{E}{\left\{ {{{\left| {\chi _{\text{AFDM}} \left( {\tau,{\nu}} \right)} \right|}^2}} \right\}_{\scriptstyle\hfill\tau  \ne 0\atop
			\scriptstyle\hfill{\rm{or}} \ \nu  \ne 0}}  \nonumber\\
	&	=	 \left( {{\mu _4} - 2} \right)\frac{1}{N}{\left| {\frac{{\sin \left[ {\pi \left( {2N{c_1}\tau  - \nu} \right)} \right]}}{{\sin \left[ {{{\pi \left( {2N{c_1}\tau  - \nu} \right)} \mathord{\left/
								{\vphantom {{\pi \left( {2N{c_1}\tau  - \nu} \right)} N}} \right.
								\kern-\nulldelimiterspace} N}} \right]}}} \right|^2} \nonumber\\
	& \quad	+ \frac{1}{N}\sum\limits_{m = 0}^{N - 1} {{{\left| {\frac{{\sin \left[ {\pi \left( {m + 2N{c_1}\tau  - \nu} \right)} \right]}}{{\sin \left[ {{{\pi \left( {m + 2N{c_1}\tau  - \nu} \right)} \mathord{\left/
										{\vphantom {{\pi \left( {m + 2N{c_1}\tau  - \nu} \right)} N}} \right.
										\kern-\nulldelimiterspace} N}} \right]}}} \right|}^2}} .
	\end{align}
	When the normalized Doppler $\nu$ is an integer, the average sidelobe level of the DPAF of AFDM without PS reduces to 
	\begin{align}\label{eq:sidelobe_DPAF_AFDM_intDop}
	\mathbb{E}{\left\{ {{{\left| {\chi _{\text{AFDM}} \left( {\tau,{\nu}} \right)} \right|}^2}} \right\}_{\scriptstyle\hfill\tau  \ne 0\atop
			\scriptstyle\hfill{\rm{or}} \ \nu  \ne 0}} =  \left( {{\mu _4} - 2} \right)N\delta \left( 2Nc_1 \tau - \nu \right) + N.
	\end{align} 
\end{corollary}
\begin{proof}
	Substituting $\tau\ne0$ or $\nu\ne0$ into (\ref{eq:AS_DPAF_AFDM_4}), (\ref{eq:sidelobe_DPAF_AFDM}) is immediately obtained. Based on this, when $\nu$ is an integer, we have
	\begin{align}
	{\frac{{\sin \left[ {\pi \left( {2N{c_1}\tau  - \nu} \right)} \right]}}{{\sin \left[ {{{\pi \left( {2N{c_1}\tau  - \nu} \right)} \mathord{\left/
							{\vphantom {{\pi \left( {2N{c_1}\tau  - \nu} \right)} N}} \right.
							\kern-\nulldelimiterspace} N}} \right]}}} = N \delta \left( 2Nc_1 \tau - \nu \right), \\
	\sum\nolimits_{m = 0}^{N - 1} {{{\left| {\frac{{\sin \left[ {\pi \left( {m + 2N{c_1}\tau  - \nu} \right)} \right]}}{{\sin \left[ {{{\pi \left( {m + 2N{c_1}\tau  - \nu} \right)} \mathord{\left/
										{\vphantom {{\pi \left( {m + 2N{c_1}\tau  - \nu} \right)} N}} \right.
										\kern-\nulldelimiterspace} N}} \right]}}} \right|}^2}} = N^2.
	\end{align}
	Consequently, we can obtain (\ref{eq:sidelobe_DPAF_AFDM_intDop}), completing the proof.
\end{proof}

With Proposition \ref{prop:AS_DPAF_noPS} and Corollaries \ref{cor:mainlobe_DPAF_AFDM} and \ref{cor:sidelobe_DPAF_AFDM} at hand, one may observe that the average squared DPAF of AFDM is independent of the parameter $c_2$. In contrary, the parameter $c_1$ influences the sidelobes of the squared DPAF of AFDM. Moreover, the kurtosis of communication data, i.e., $\mu _4$, contributes to both the mainlobe and the sidelobes. Next, we present an in-depth and systematic examination of the impacts of the parameter $c_1$ and $\mu _4$ on the sidelobes. While this examination focuses on the case of integer normalized Doppler, the conclusions also hold for the case of fractional normalized Doppler, as verified by example \ref{ex:th_DPAF_AFDM}.

\begin{proposition} \label{prop:solution}
	When $2Nc_1$ is an integer and a sub-Gaussian constellation is employed, in the case of integer normalized Doppler, the average sidelobes of squared DPAF of AFDM without PS take only two possible values, i.e., $N$ (corresponding to the ``sea level'' following Lemma \ref{lemma:Iceberg}) and $\left( {{\mu _4} - 1} \right)N$ (referred to as the depression), according to (\ref{eq:sidelobe_DPAF_AFDM_intDop}). Moreover, the number of depressions is $N-1$, whose positions in the delay-Doppler plane are given by $\left(\tau,{{\left\langle 2N{c_1}\tau \right\rangle }_N}\right), \tau\in\left[1,N-1\right]$.	The delay gap and Doppler gap between adjacent depressions are equal to $\frac{N}{2Nc_1}$ and $2Nc_1$, respectively. An illustrative example is shown in Fig. \ref{fg:example_th_DPAF_AFDM}(a) within Example \ref{ex:th_DPAF_AFDM}. 
\end{proposition}
\begin{proof}
	See Appendex \ref{proof_prop_solution}.
\end{proof}

Similarly, the sidelobes of the squared DPAFs of OFDM and OCDM waveforms can be obtained as follows.

\begin{corollary} \label{cor:sidelobe_DPAF_OFDM}
	The average sidelobe level of the squared DPAF of OFDM without PS is 	
	\begin{align}\label{eq:sidelobe_DPAF_OFDM}
	&\mathbb{E}{\left\{ {{{\left| {\chi _{\text{OFDM}} \left( {\tau,{\nu}} \right)} \right|}^2}} \right\}_{\scriptstyle\hfill\tau  \ne 0\atop
			\scriptstyle\hfill{\rm{or}} \ \nu  \ne 0}}  	=	 \left( {{\mu _4} - 2} \right)\frac{1}{N}{\left| {\frac{{\sin \left( {\pi {\nu}} \right)}}{{\sin \left( {{{\pi {\nu}} \mathord{\left/
								{\vphantom {{\pi \left( {\nu} \right)} N}} \right.
								\kern-\nulldelimiterspace} N}} \right)}}} \right|^2} \nonumber\\
	& \quad	+ \frac{1}{N}\sum\limits_{m = 0}^{N - 1} {{{\left| {\frac{{\sin \left[ {\pi \left( {m - \nu} \right)} \right]}}{{\sin \left[ {{{\pi \left( {m - \nu} \right)} \mathord{\left/
										{\vphantom {{\pi \left( {m - \nu} \right)} N}} \right.
										\kern-\nulldelimiterspace} N}} \right]}}} \right|}^2}} .
	\end{align}
	When the normalized Doppler $\nu$ is an integer, the average sidelobe level of the squared DPAF of OFDM without PS reduces to 	
	\begin{align}\label{eq:sidelobe_DPAF_OFDM_intDop}
	\mathbb{E}{\left\{ {{{\left| {\chi _{\text{OFDM}} \left( {\tau,{\nu}} \right)} \right|}^2}} \right\}_{\scriptstyle\hfill\tau  \ne 0\atop
			\scriptstyle\hfill{\rm{or}} \ \nu  \ne 0}} =  \left( {{\mu _4} - 2} \right)N\delta \left( \nu \right) + N.
	\end{align} 
\end{corollary}

\begin{corollary} \label{cor:sidelobe_DPAF_OCDM}
	The average sidelobe level of the squared DPAF of OCDM without PS is 	
	\begin{align}\label{eq:sidelobe_DPAF_OCDM}
	&\mathbb{E}{\left\{ {{{\left| {\chi _{\text{OCDM}} \left( {\tau,{\nu}} \right)} \right|}^2}} \right\}_{\scriptstyle\hfill\tau  \ne 0\atop
			\scriptstyle\hfill{\rm{or}} \ \nu  \ne 0}}  	=	 \left( {{\mu _4} - 2} \right)\frac{1}{N}{\left| {\frac{{\sin \left[ {\pi \left( {\tau  - \nu} \right)} \right]}}{{\sin \left[ {{{\pi \left( {\tau  - \nu} \right)} \mathord{\left/
								{\vphantom {{\pi \left( {\tau  - \nu} \right)} N}} \right.
								\kern-\nulldelimiterspace} N}} \right]}}} \right|^2} \nonumber\\
	& \quad	+ \frac{1}{N}\sum\limits_{m = 0}^{N - 1} {{{\left| {\frac{{\sin \left[ {\pi \left( {m + \tau  - \nu} \right)} \right]}}{{\sin \left[ {{{\pi \left( {m + \tau  - \nu} \right)} \mathord{\left/
										{\vphantom {{\pi \left( {m + \tau  - \nu} \right)} N}} \right.
										\kern-\nulldelimiterspace} N}} \right]}}} \right|}^2}} .
	\end{align}
	When the normalized Doppler $\nu$ is an integer, the average sidelobe level of the squared DPAF of OCDM without PS reduces to 
	\begin{align}\label{eq:sidelobe_DPAF_OCDM_intDop}
	\mathbb{E}{\left\{ {{{\left| {\chi _{\text{OCDM}} \left( {\tau,{\nu}} \right)} \right|}^2}} \right\}_{\scriptstyle\hfill\tau  \ne 0\atop
			\scriptstyle\hfill{\rm{or}} \ \nu  \ne 0}} =  \left( {{\mu _4} - 2} \right)N\delta \left(  \tau - \nu \right) + N.
	\end{align} 
\end{corollary}

{\bf{Comparison among DPAFs of AFDM, OFDM and OCDM}}: According to Corollaries \ref{cor:sidelobe_DPAF_AFDM} to \ref{cor:sidelobe_DPAF_OCDM}, we note that there are $N-1$ depressions in the sidelobes of the average squared DPAF for all three waveforms, which are however positioned at different delays and Doppler indices for AFDM, OFDM and OCDM. Specifically, the indices of depressions are fixed for OFDM and OCDM. For the DPAF of OFDM, all $N-1$ depressions are positioned in the delay cut, i.e., the indices $\left(\tau,0\right)$, $\tau\in \left[1,N-1\right]$. For the DPAF of OCDM, the $N-1$ depressions are positioned at the indices $\left(\tau,\tau\right)$, $\tau\in \left[1,N-1\right]$. For the DPAF of AFDM, the $N-1$ depressions are located in the 2D delay-Doppler plane with the indices $\left(\tau,{{\left\langle 2N{c_1}\tau \right\rangle }_N}\right), \tau\in\left[1,N-1\right]$, which can be flexibly adjusted by properly tuning the parameter $c_1$. 

To intuitively illustrate the DPAFs of AFDM, OFDM and OCDM, the following example is presented for the case of fractional normalized Doppler.
\begin{example}\label{ex:th_DPAF_AFDM}
	When $N=128$, 16-QAM is used, and $2Nc_1 = 8$ for AFDM, the 2D contour plots of derived theoretical squared DPAFs of AFDM, OFDM and OCDM waveforms for the case of fractional normalized Doppler are illustrated in Fig. \ref{fg:example_th_DPAF_AFDM}. For the DPAF of AFDM, we observe that there are $N-1=127$ depressions within the range $\tau, \nu \in \left[-{N \mathord{\left/
			{\vphantom {N 2}} \right.
			\kern-\nulldelimiterspace} 2},{N \mathord{\left/
			{\vphantom {N 2}} \right.
			\kern-\nulldelimiterspace} 2}-1\right]$. Moreover, the delay gap between adjacent depressions is $\frac{N}{2Nc_1}=16$, and the Doppler gap between adjacent depressions is $2Nc_1=8$. These results are consistent with the conclusion in Proposition \ref{prop:solution}. In contrast to the case of integer normalized Doppler, there is a smooth transition region from the ``sea level'' to the depression for the case of fractional normalized Doppler. For the DPAFs of OFDM, all $N-1$ depressions are positioned along the delay cut of the squared DPAF. The depressions of squared DPAF of OCDM are located at the index $\left(\tau,\tau\right)$, $\tau\in \left[-{N \mathord{\left/
			{\vphantom {N 2}} \right.
			\kern-\nulldelimiterspace} 2},-1\right]  \cup \left[1,{N \mathord{\left/
			{\vphantom {N 2}} \right.
			\kern-\nulldelimiterspace} 2}-1\right]$, which also matches the above discussion.
	\begin{figure*}[!htbp]
		\vspace*{-10pt} 
		\centering
		\subfigure[DPAF of AFDM]{
			\includegraphics[width=2.25in]{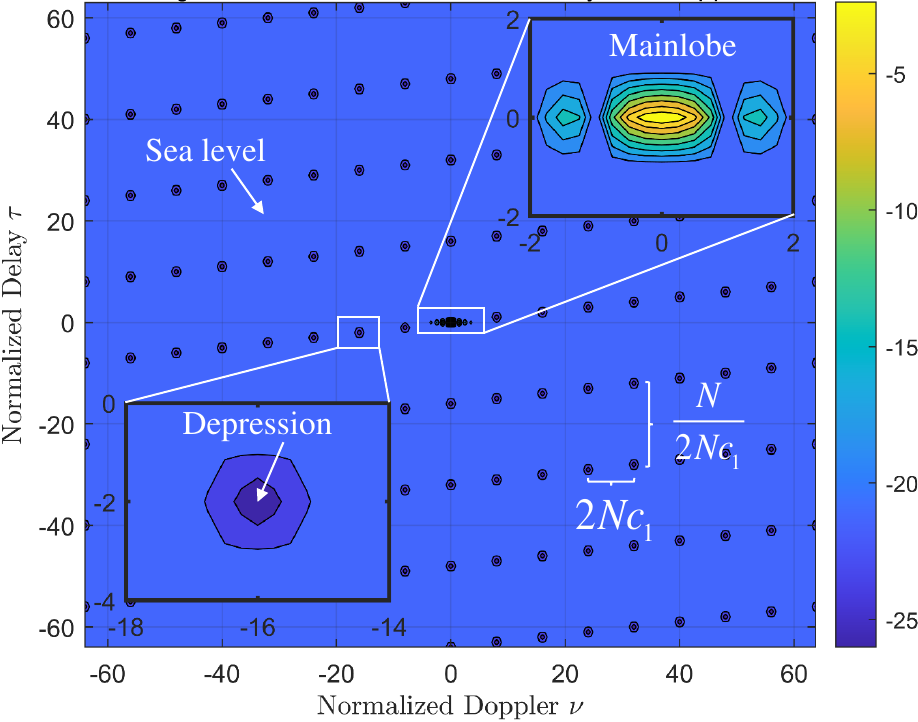}			
		}	
		\subfigure[DPAF of OFDM]{
			\includegraphics[width=2.25in]{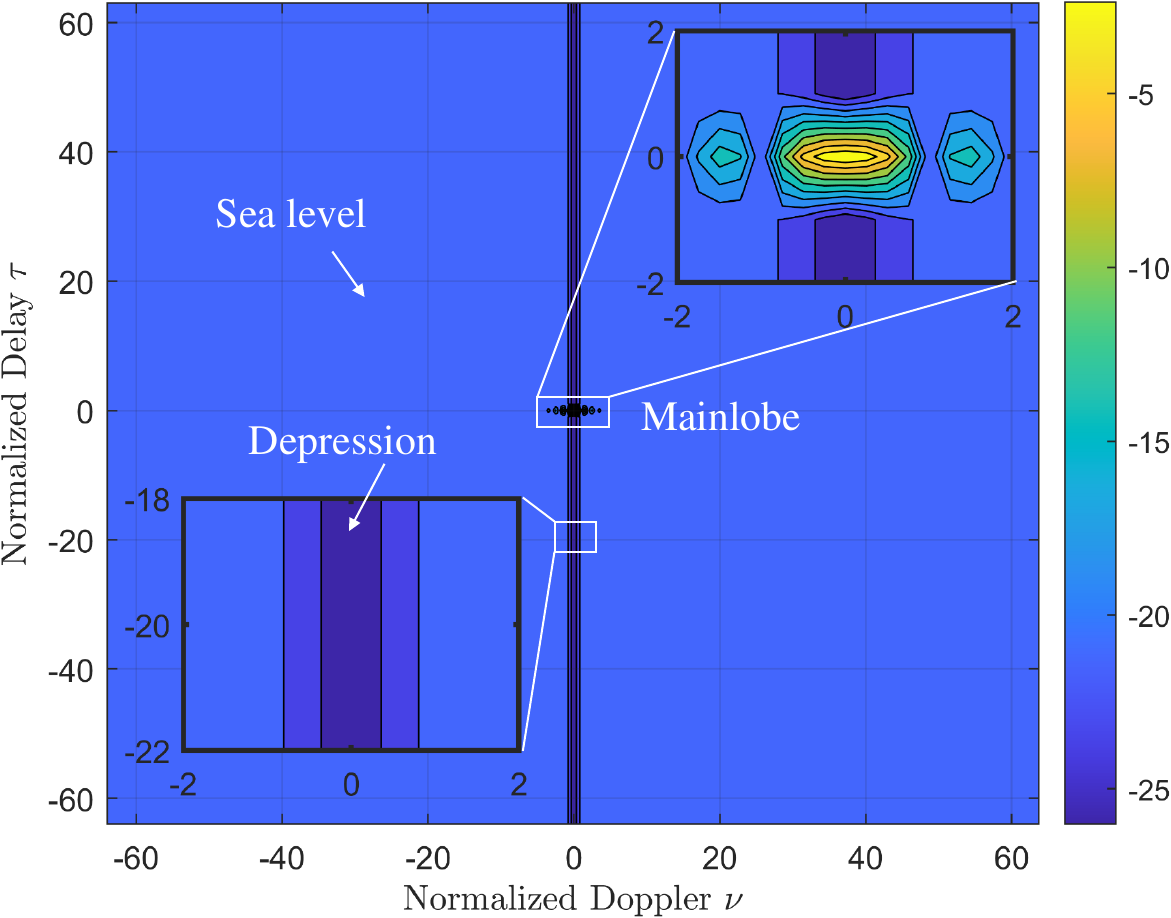}
		}
		\subfigure[DPAF of OCDM]{
			\includegraphics[width=2.25in]{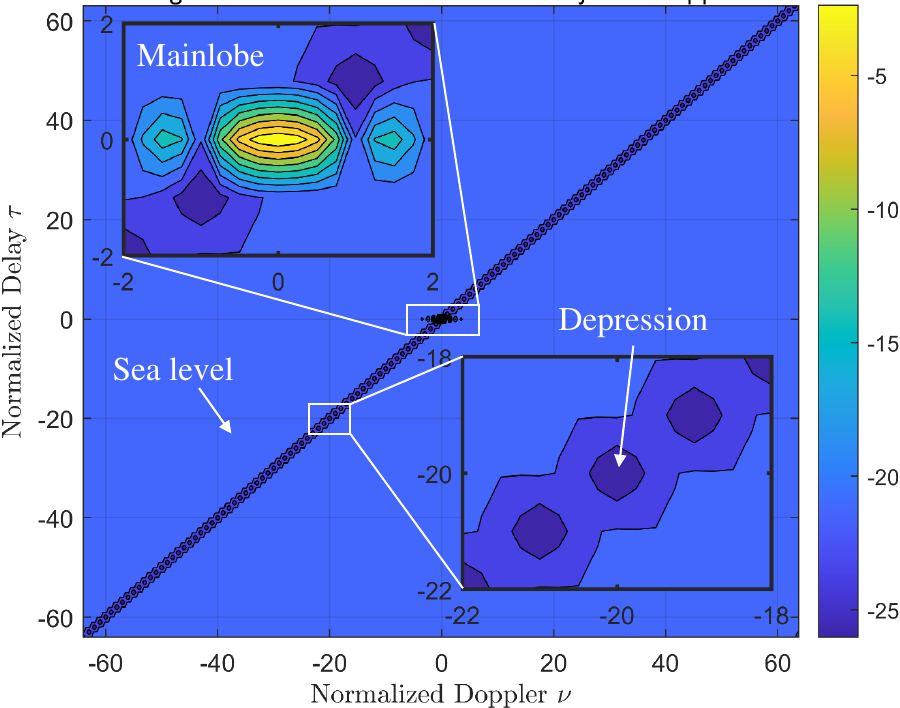}
		}
		\caption{The 2D Contour plots of theoretical derived DPAFs of AFDM, OFDM and OCDM without PS for the case of fractional normalized Doppler, where $N=128$, $2Nc_1=8$ and 16-QAM.  
			\label{fg:example_th_DPAF_AFDM}}
			\vspace*{-10pt} 
	\end{figure*}
\end{example}

\subsection{DPAF-Inspired Design Guideline for AFDM Parameter}


\subsubsection{Impact of Sidelobes of DPAF on Sensing Parameter Estimation}

\begin{figure}[!htbp]
	\centering	
	\subfigure[Weak target at depression part]{
		\includegraphics[width=1.62in]{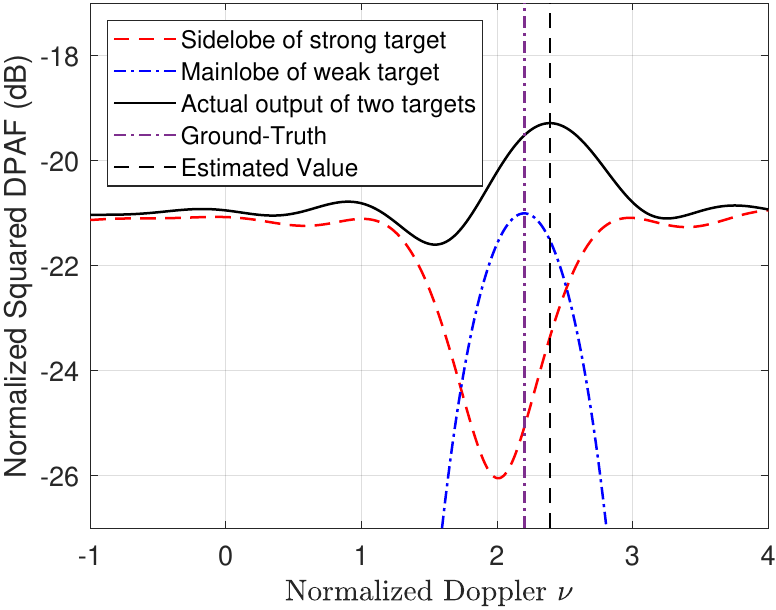}
	}
	\subfigure[Weak target at ``sea level'' part]{
		\includegraphics[width=1.62in]{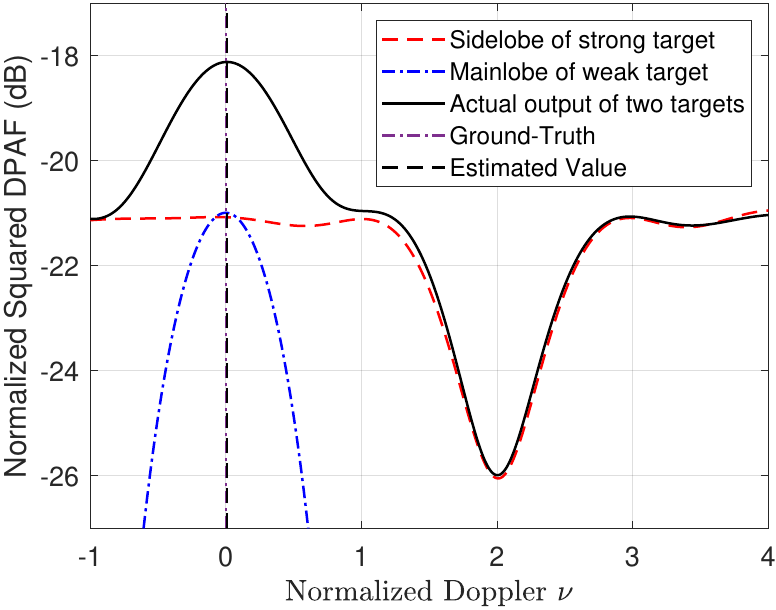}
	}		
	\caption{Simulated output of noncoherent integration after matched filter in a strong-weak target scenario with $N=128$ and 16-QAM, where the amplitude of the strong target is 21 dB higher than that of the weak target.  
		\label{fg:DPAF_strongWeakTarget}}
\end{figure}

We first analyze the impact of the sidelobes of the DPAF on sensing parameter estimation in a strong-weak target scenario, where there are two Swerling 2 targets, and the power of the echo from the strong target is significantly higher than that from the weak target. According to (\ref{eq:ave_squ_rk}), the actual output of noncoherent integration is the sum of the average squared DPAFs corresponding to the two targets. 
As shown in Fig. \ref{fg:DPAF_strongWeakTarget}, if the mainlobe of the weak target is located at the depression part of the strong target sidelobes, the shape of acutual output correspongding the mainlobe of the weak target may become distorted, due to the uneven impact from the depressions. The distorted mainlobe may lead to a significantly erroneous Doppler estimate for the weak target, even if the signal-to-noise ratio (SNR) is high. In contrast, if the mainlobe of the weak target is positioned at the ``sea level'' part of the strong target sidelobes, the amplitude of the mainlobe of the weak target is overall raised. As a result, the estimated Doppler of the weak target can still match closely the ground truth. This suggests that the depression imposes a negative effect on the parameter estimation of the weak target in the presence of the strong target. Hence, one should avoid positioning the weak target in the depressions of the strong target sidelobes, which motivates the following design guideline for AFDM parameter. 

\subsubsection{Design Guideline for AFDM Parameter} According to the comparison in Sec. \ref{sec:Analy_DPAF}, the AFDM waveform can flexibly control the positions of depressions by adjusting the parameter $c_1$. Based on these insights, we propose the following design guideline for the AFDM parameter $c_1$.
\begin{corollary}\label{cor:c1_design}
	For a strong-weak target scenario, let the ranges of the strong and weak targets be denoted by $d_{\rm s}$ and $d_{\rm w}$, and their velocities be denoted by $v_{\rm s}$ and $v_{\rm w}$, respectively. To avoid placing the weak target at the depressions of the strong target sidelobes, the AFDM parameter $c_1$ should be chosen such that ${c_1} \ne {\bar c_1}$, where
	\begin{align}\label{eq:c1_design}
	{\bar c_1} \buildrel \Delta \over =  \frac{{c\left( {{v_{\rm w}} - {v_{\rm s}}} \right){f_c}}}{{2\left( {{d_{\rm w}} - {d_{\rm s}}} \right)f_s^2}} + k\frac{c}{{4\left( {{d_{\rm w}} - {d_{\rm s}}} \right){f_s}}}, \ k \in \mathbb{Z},
	\end{align}	
	with $c$ denoting the speed of light and $f_c$ being the carrier frequency.
\end{corollary}
\begin{proof}
	See Appendex \ref{proof:c1_design}.
\end{proof}
\begin{remark}
	Once the sensing targets have been tracked, their ranges and velocities can be roughly obtained through a forward prediction process including prediction errors. To enhance robustness against the prediction errors, $c_1$ could be set such that ${c_1} \notin \left[ {{{\bar c}_1} - {\sigma _c},{{\bar c}_1} + {\sigma _c}} \right]$, where ${\sigma _c}$ is affected by the prediction errors. Moreover, while the design guideline is proposed for the AFDM without PS, numerical results in Sec. \ref{sec:simulation} demonstrate its applicability to pulse-shaped AFDM.
\end{remark}

\section{Closed-Form Expression for Average Squared DPAF of Pulse-Shaped AFDM}

This section presents the analytical derivation of the closed-form expression for the average squared DPAF of pulse-shaped random AFDM. Based on this, we reveal the impact of PS on the AF.

\subsection{Closed-Form Expression of DPAF of Pulse-Shaped AFDM}

According to (\ref{eq:AFDM}) and (\ref{eq:pulsed_AFDM}), the pulse-shaped AFDM symbol can be expressed as
\begin{align} \label{eq:pulsed_AFDM3}
 {{x}_{{\rm ps}}}\left[i\right] &= \sum\nolimits_{n = 0}^{N - 1} {g_{{{\left\langle {i - nL} \right\rangle }_{NL}}} {x}_n}  \nonumber\\ 
& = \frac{1}{{\sqrt N }}\sum\limits_{n = 0}^{N - 1} {\sum\limits_{m = 0}^{N - 1} {g_{{{\left\langle {i - nL} \right\rangle }_{NL}}} {s}_ m{\phi _n}\left( m \right)} } ,
\end{align} 
where $i \in \left[0,NL-1\right]$, and $\mathbf{g}\in \mathbb{R}^{NL\times 1}$ is the first column of the matrix $\mathbf{G}$, denoting the effective PS response, defined as
\begin{align}
\mathbf{g} = \left[{\tilde g_{ML}}, \cdots, {\tilde g_{2ML}}, \mathbf{0}^{T}_{NL-2ML-1}, {\tilde g_0}, \cdots ,{\tilde g_{ML-1}}\right]^{{T}}.
\end{align} 
Consequently, the average squared DPAF of pulse-shaped AFDM can be formulated as (\ref{eq:AS_DPAF_pulsed_1}) at the top of this page.

\begin{figure*} [htb]
	\vspace*{-10pt} 
	\small
	\begin{flalign}\label{eq:AS_DPAF_pulsed_1}
	&\	\mathbb{E} \left\{ {{{\left| {\chi _{\text{AFDM,PS}} \left( {\tau ,\nu } \right)} \right|}^2}} \right\} = \frac{1}{{{N^2}}}\sum\limits_{n = 0}^{N - 1} {\sum\limits_{n' = 0}^{N - 1} {\sum\limits_{m = 0}^{N - 1} {\sum\limits_{m' = 0}^{N - 1} {\sum\limits_{q = 0}^{N - 1} {\sum\limits_{q' = 0}^{N - 1} {\sum\limits_{p = 0}^{N - 1} {\sum\limits_{p' = 0}^{N - 1} {\sum\limits_{k = 0}^{NL - 1} {\sum\limits_{k' = 0}^{NL - 1} \Bigg\{ {\mathbb{E}\left\{ {{{\bar s}_m^{*}} \bar s_{m'}^{} \bar s_p^{} {{\bar s}_{p'}^*}} \right\}} {e^{ - j2\pi {c_1}\left( {{n^2} - {{n'}^2}} \right)}}{e^{ - j\frac{{2\pi }}{N}\left( {mn - m'n'} \right)}}   } } } } } } } } } & \nonumber\\
	&\ \quad\quad\quad\quad\quad\quad\quad\quad\quad\quad\quad \cdot	{e^{j2\pi {c_1}\left( {{q^2} - {{q'}^2}} \right)}}{e^{j\frac{{2\pi }}{N}\left( {pq - p'q'} \right)}}  {g}_{{{\left\langle {k - nL} \right\rangle }_{NL}}}  g_{{{\left\langle {k - n'L - \tau } \right\rangle }_{NL}}}  g_{{{\left\langle {k' - qL} \right\rangle }_{NL}}} {g}_{{{\left\langle {k' - q'L - \tau } \right\rangle }_{NL}}}{e^{ j\frac{{2\pi }}{{NL}}\nu \left( {k - k'} \right)}} \Bigg\} .& 
	\end{flalign}
	\normalsize
	\hrule
	\vspace*{-10pt} 
\end{figure*}

Next, the closed-form expression for the average squared DPAF of pulse-shaped AFDM is provided by the following proposition. Although the closed-form expression in (\ref{eq:AS_DPAF_pulsed_2}) is derived for the case of integer normalized Doppler, numerical results in Sec. \ref{sec:simulation} validate that the derived closed-form expression still works for the case of fractional normalized Doppler.

\begin{proposition} \label{prop:AS_DPAF_pulsed}
	When the normalized Doppler $\nu$ is an integer, the closed-form expression for the average squared DPAF of pulse-shaped AFDM is given by
	\begin{flalign} \label{eq:AS_DPAF_pulsed_2}
	&\ \mathbb{E}\left\{ {{{\left| {\chi _{\text{AFDM,PS}} \left( {\tau ,\nu} \right)} \right|}^2}} \right\} & \nonumber \\
	&\  = {\left| {\frac{{\sin \left( {\pi \nu} \right)}}{{\sin \left( {{{\pi \nu} \mathord{\left/
								{\vphantom {{\pi \nu} N}} \right.
								\kern-\nulldelimiterspace} N}} \right)}}} \right|^2}{\left| {{\chi _g}\left( {\tau ,\nu} \right)} \right|^2} + \sum\limits_{n = 0}^{N - 1} { \Bigg\{ {{\left| {{\chi _g}\left( {{{\left\langle {\tau  - nL} \right\rangle }_{NL}},\nu} \right)} \right|}^2} } & \nonumber \\ 
	&\ \quad \cdot {\left[ {\frac{{\left( {{\mu _4} - 2} \right)}}{N}{{\left| {\frac{{\sin \left[ {\pi \left( {2N{c_1}n  - \nu} \right)} \right]}}{{\sin \left[ {{{\pi \left( {2N{c_1}n  - \nu} \right)} \mathord{\left/
											{\vphantom {{\pi \left( {2N{c_1}n  - \nu} \right)} N}} \right.
											\kern-\nulldelimiterspace} N}} \right]}}} \right|}^2} + N} \right]} \Bigg\}, & 			
	\end{flalign} 
	where ${\chi _g}$ is the DPAF of $\mathbf{g}$, referred to as DPAF-g, i.e.,
	\begin{equation}
	{\chi _g}\left( {\tau ,\nu} \right) = \sum\nolimits_{m = 0}^{NL - 1} {g_m  {g}_{{{\left\langle {m - \tau } \right\rangle }_{NL}}} {e^{-j\frac{{2\pi }}{{NL}}\nu m}}}.
	\end{equation}
\end{proposition}
\begin{proof}
	See Appendex \ref{proof:AS_DPAF_pulsed}.	
\end{proof}

\subsection{Discussion on DPAF of Pulse-Shaped AFDM}


\subsubsection{Mainlobe and delay/Doppler cut of DPAF of pulse-shaped AFDM}

\begin{corollary} \label{cor:mainlobe_DPAF_pulsed_AFDM}
	The average mainlobe level of squared DPAF of pulse-shaped AFDM is given by	
	\begin{align}\label{eq:mainlobe_DPAF_pulsed_AFDM}
	& \mathbb{E}\left\{ {{{\left| {\chi _{\text{AFDM,PS}} \left( {0,0} \right)} \right|}^2}} \right\} =  {N^2} + N\sum\limits_{n = 0}^{N - 1} {{{\left| {{R_g}\left( {{{\left\langle { - nL} \right\rangle }_{NL}}} \right)} \right|}^2}}  \nonumber \\
	& + \left( {{\mu _4} - 2} \right)N\sum\limits_{n = 0}^{2N{c_1} - 1} {{{\left| {{R_g}\left( {{{\left\langle { - n\frac{{NL}}{{2N{c_1}}}} \right\rangle }_{NL}}} \right)} \right|}^2}} ,
	\end{align} \normalsize
	where ${R_g}$ is the PACF of effective PS response $\mathbf{g}$, defined as
	\begin{equation}
	{R_g}\left( \tau  \right) = \sum\nolimits_{m = 0}^{NL - 1} {g_m  g_{{{\left\langle {m - \tau } \right\rangle }_{NL}}} }.
	\end{equation}	
\end{corollary}
\begin{proof}
	When $\tau=0$ and $\nu=0$, we have 
	\begin{equation}
	{\left| {\frac{{\sin \left[ {\pi \left( {2N{c_1}n} \right)} \right]}}{{\sin \left[ {{{\pi \left( {2N{c_1}n} \right)} \mathord{\left/
								{\vphantom {{\pi \left( {2N{c_1}n} \right)} N}} \right.
								\kern-\nulldelimiterspace} N}} \right]}}} \right|^2} = \left\{ {\begin{array}{*{20}{c}}
		{{N^2},}&{n = {{\beta N} \mathord{\left/
					{\vphantom {{\beta N} {\left( {2N{c_1}} \right)}}} \right.
					\kern-\nulldelimiterspace} {\left( {2N{c_1}} \right)}},}\\
		{0,}&{otherwise,}
		\end{array}} \right.
	\end{equation}
	where $\beta  \in \left[ {0,2N{c_1}} \right]$ is an integer. Consequently, it immediately yields (\ref{eq:mainlobe_DPAF_pulsed_AFDM}), thereby completing the proof.
\end{proof}

Next, we present the delay cut and Doppler cut of the average squared DPAF of pulse-shaped AFDM.

\begin{corollary} \label{cor:delayCut_DPAF_AFDM_Pulsed}
	The average delay cut of the squared DPAF of pulse-shaped AFDM is expressed as	
	\begin{flalign} \label{eq:delayCut_DPAF_pulsed_1}
	&\ \mathbb{E}\left\{ {{{\left| {\chi _{\text{AFDM,PS}} \left( {\tau ,0} \right)} \right|}^2}} \right\} \hspace{-0.125cm} = \hspace{-0.125cm} {N^2}{\left| {{R_g}\left( \tau  \right)} \right|^2} \hspace*{-1pt} \hspace{-0.125cm} + \hspace{-0.125cm} \sum\limits_{n = 0}^{N - 1} { \hspace{-0.125cm} \Bigg\{ \hspace{-0.125cm} {{\left| {{R_g}\left( {{{\left\langle {\tau  - nL} \right\rangle }_{NL}}} \right)} \right|}^2} } & \nonumber\\
	&\ \quad\quad \cdot	\left[ {\frac{{\left( {{\mu _4} - 2} \right)}}{N}{{\left| {\frac{{\sin \left[ {\pi \left( {2N{c_1}n } \right)} \right]}}{{\sin \left[ {{{\pi \left( {2N{c_1}n } \right)} \mathord{\left/
										{\vphantom {{\pi \left( {2N{c_1}n } \right)} N}} \right.
										\kern-\nulldelimiterspace} N}} \right]}}} \right|}^2} \hspace*{-3pt} + \hspace*{-2pt} N} \right] \Bigg\}. &
	\end{flalign} 
\end{corollary}
\begin{proof}
	Substituting $\nu=0$ into (\ref{eq:AS_DPAF_pulsed_2}) immediately yields (\ref{eq:delayCut_DPAF_pulsed_1}), thereby completing the proof.
\end{proof}

\begin{corollary} \label{cor:DopplerCut_DPAF_AFDM_Pulsed}
	The average Doppler cut of the squared DPAF of pulse-shaped AFDM is given by	
	\begin{flalign} \label{eq:DopplerCut_DPAF_pulsed_1}
	&\ \mathbb{E}\left\{ {{{\left| {\chi _{\text{AFDM,PS}} \left( {0,\nu} \right)} \right|}^2}} \right\} & \nonumber\\
	&\ = {\left| {\frac{{\sin \left( {\pi \nu} \right)}}{{\sin \left( {{{\pi \nu} \mathord{\left/
								{\vphantom {{\pi \nu} N}} \right.
								\kern-\nulldelimiterspace} N}} \right)}}} \right|^2}{\left| \mathcal{{F}}_{g}\left(\nu\right) \right|^2} + \sum\limits_{n = 0}^{N - 1} {\Bigg\{ {{\left| {{\chi _g}\left( {{{\left\langle { - nL} \right\rangle }_{NL}},\nu} \right)} \right|}^2} }& \nonumber\\
	&\ \quad \cdot {\left[ {\frac{{\left( {{\mu _4} - 2} \right)}}{N}{{\left| {\frac{{\sin \left[ {\pi \left( {2N{c_1}n - \nu} \right)} \right]}}{{\sin \left[ {{{\pi \left( {2N{c_1}n - \nu} \right)} \mathord{\left/
											{\vphantom {{\pi \left( {2N{c_1}n - \nu} \right)} N}} \right.
											\kern-\nulldelimiterspace} N}} \right]}}} \right|}^2} + N} \right]} \Bigg\}, & 
	\end{flalign} \normalsize
	where $\mathcal{{F}}_{g}$ denotes the spectrum of squared envelope (SSE) of $\mathbf{g}$, defined as
	\begin{equation}
	\mathcal{{F}}_{g}\left(\nu\right)={\sum\nolimits_{m = 0}^{NL - 1} {{{\left| {g_m} \right|}^2}{e^{-j\frac{{2\pi }}{{NL}} \nu m}}} }.
	\end{equation}		
\end{corollary}
\begin{proof}
	Substituting $\tau=0$ into (\ref{eq:AS_DPAF_pulsed_2}) immediately yields (\ref{eq:DopplerCut_DPAF_pulsed_1}), thereby completing the proof.
\end{proof}

\subsubsection{Impact analysis of PS on DPAF}

To intuitively understand the DPAF of pulse-shaped AFDM, we present the following example.
\begin{example}\label{ex:impact_fulse_DPAF_AFDM}
	Given $N=128$, $2Nc_1 = 8$, $L=4$, 16-QAM and a PS filter being the well-known RRC filter with $\alpha = 0.35$, the resulting theoretical average squared DPAF of pulse-shaped AFDM is illustrated in Figs. \ref{fg:th_DPAF_AFDM_pulsed} and \ref{fg:th_impact_fulse_cuts_DPAF_AFDM}. 
	\begin{figure}[!htbp]
		\centering
		\subfigure[Theoretical DPAF of AFDM]{
			\includegraphics[width=1.6in]{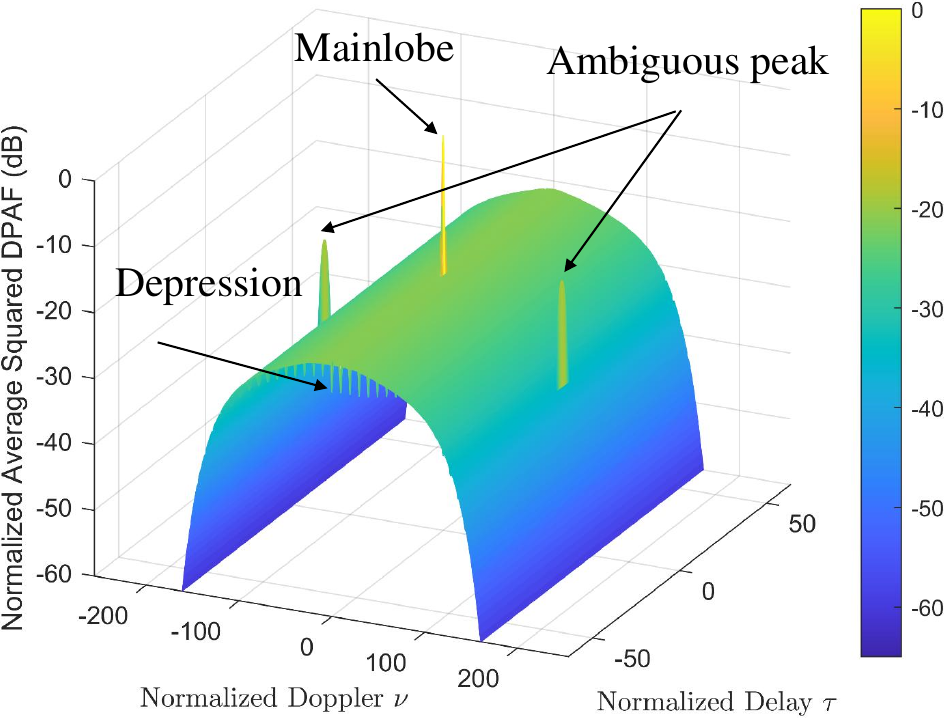}
		}
		\subfigure[2D Contour plots]{
			\includegraphics[width=1.6in]{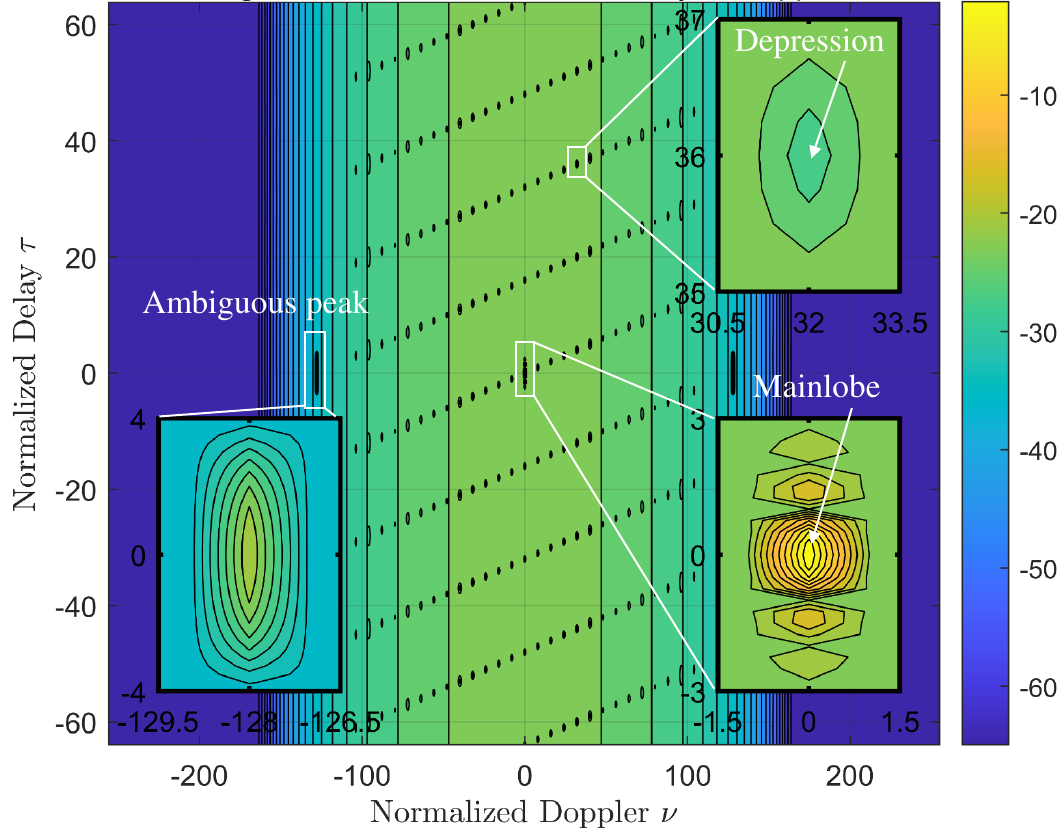}
		}	
		\caption{The theoretical average squared DPAF of pulse-shaped AFDM for $N=128$, $2Nc_1 = 8$, $L=4$, and 16-QAM.  
			\label{fg:th_DPAF_AFDM_pulsed}}
		\vspace*{-5pt} 
	\end{figure}
	\begin{figure}[!htbp]
		\centering
		\subfigure[Delay Cut]{
			\includegraphics[width=2.3in]{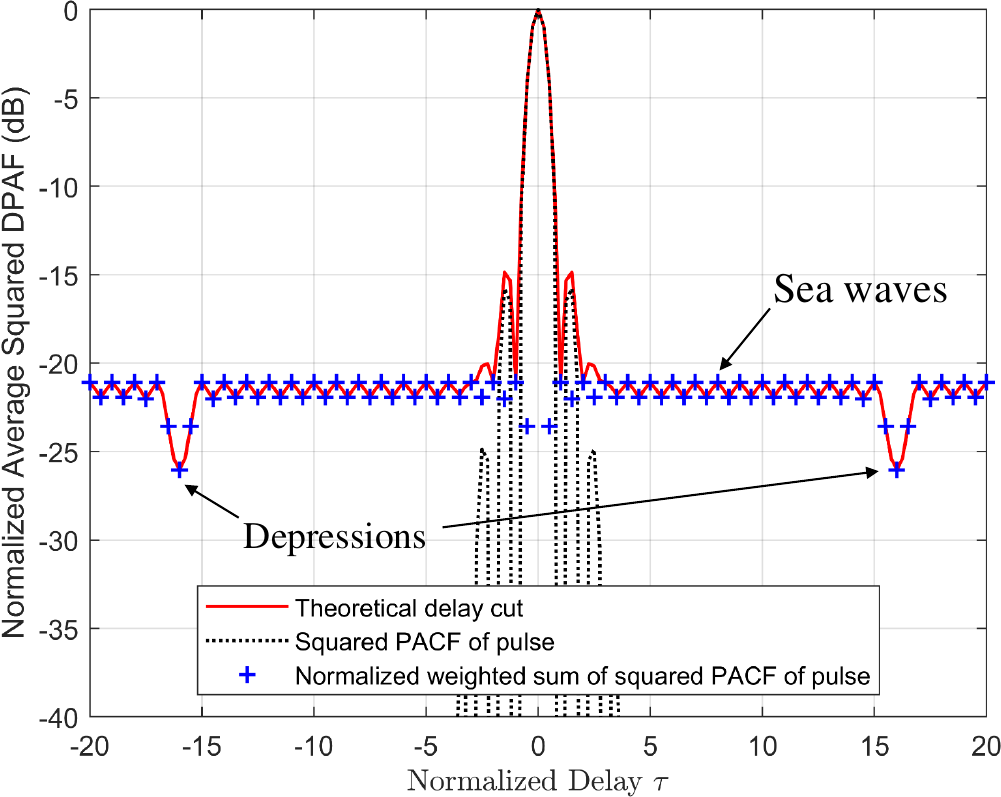}
		}
		\subfigure[Doppeler Cut]{
			\includegraphics[width=2.3in]{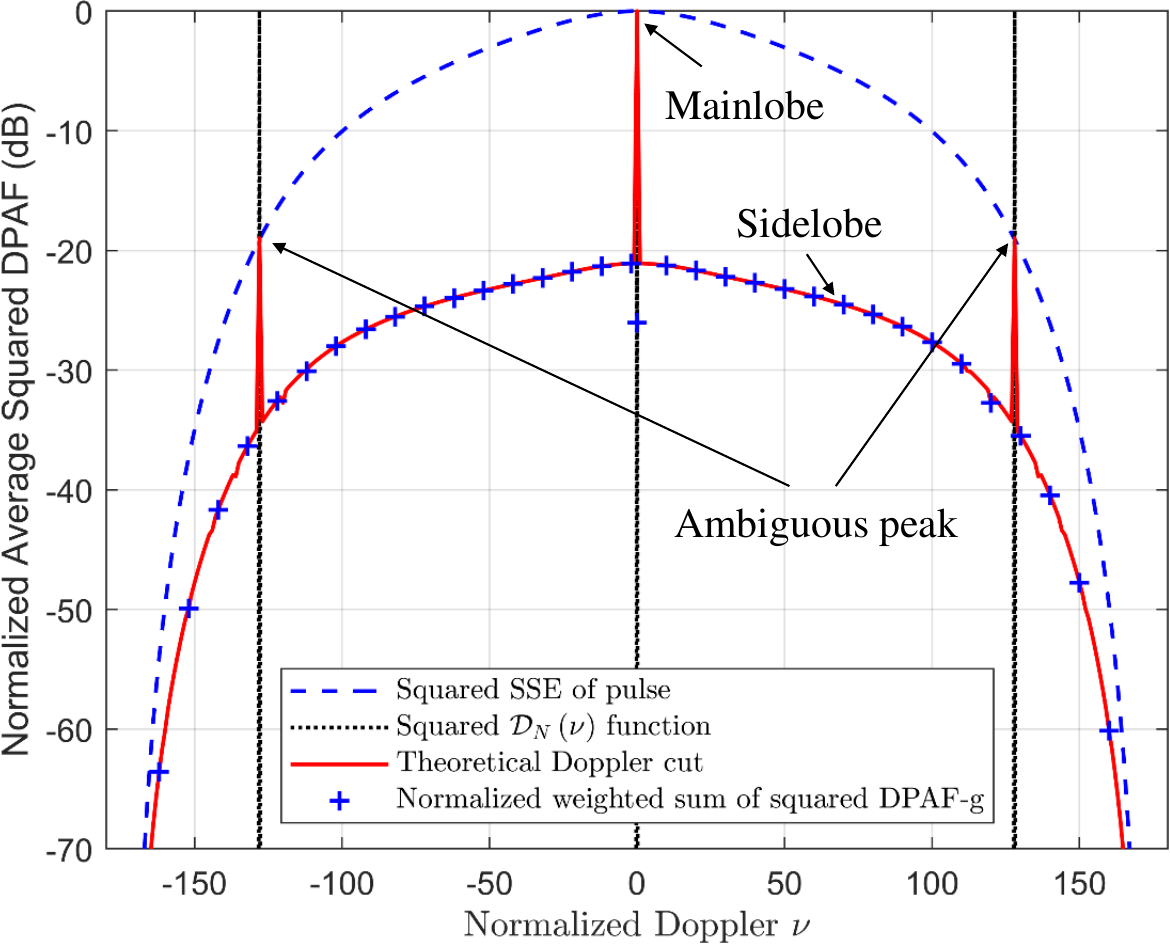}
		}	
		\caption{The theoretical delay cut and Doppler cut of derived DPAF of pulse-shaped AFDM for $N=128$, $2Nc_1 = 8$, $L=4$, and 16-QAM.  
			\label{fg:th_impact_fulse_cuts_DPAF_AFDM}}
		\vspace*{-5pt} 
	\end{figure}
\end{example}

It can be observed that compared with the DPAF of AFDM without PS shown in Fig. \ref{fg:example_th_DPAF_AFDM}, the mainlobe of the DPAF of pulse-shaped AFDM is shaped along the delay axis, while the sidelobe rolls off rapidly along the Doppler axis, as shown in Fig. \ref{fg:th_DPAF_AFDM_pulsed}. Moreover, there are two ambiguous peaks located at $\left(-N,0\right)$ and $\left(N,0\right)$, whose powers are lower than the mainlobe but higher than the sidelobes, which are influenced by both up-sampling and PS. Meanwhile, while the impact of PS exists, the positions of depressions in DPAF of AFDM remain unchanged compared to Fig. \ref{fg:example_th_DPAF_AFDM}.

More intuitive impacts of PS on the DPAF are illustrated in Fig. \ref{fg:th_impact_fulse_cuts_DPAF_AFDM}. On the one hand, as shown in Fig. \ref{fg:th_impact_fulse_cuts_DPAF_AFDM}(a), the delay cut of the DPAF of pulse-shaped AFDM closely matches the squared PACF of the PS response $\mathbf{g}$ within the delay region $\left[-1,1\right]$, corresponding to the ``iceberg" component of DPAF following Lemma \ref{lemma:Iceberg}. This means that the mainlobe of the delay cut is primarily contributed by the squared PACF of $\mathbf{g}$, which is the first term on the right-hand side of Eq. (\ref{eq:delayCut_DPAF_pulsed_1}). Beyond this region, the delay cut closely matches the weighted sum of squared PACF of $\mathbf{g}$, which is the second term on the right-hand side of Eq. (\ref{eq:delayCut_DPAF_pulsed_1}). This suggests that the sidelobe level is mainly affected by the weighted sum of squared PACF of $\mathbf{g}$, attributed to the ``sea level" with ``sea waves''. 

On the other hand, as shown in Fig. \ref{fg:th_impact_fulse_cuts_DPAF_AFDM}(b), for the Doppler cut of DPAF of pulse-shaped AFDM, the peaks including the main peak and two ambiguous peak are contributed by the squared Dirichlet function $\mathcal{{D}}_{N}\left(\nu\right)$ and the squared SSE of $\mathbf{g}$, simultanously, corresponding to the first term on the right-hand side of Eq. (\ref{eq:DopplerCut_DPAF_pulsed_1}). The Dirichlet function is a periodic function with a period of $N$, and the SSE of $\mathbf{g}$ exhibits the rapid roll-off, which is why two ambiguous peaks have lower power than the main peak. Moreover, the sidelobes of DPAF of pulse-shaped AFDM closely match the weighted sum of squared DPAF-g along the Doppler axis, corresponding to the second term on the right-hand side of Eq. (\ref{eq:DopplerCut_DPAF_pulsed_1}), which also exhibits the rapid roll-off.

\section{Simulation Results} \label{sec:simulation}
In this section, numerical results based on Monte Carlo simulations are presented. 
In accordance with \cite{liu2024ofdm}, we employ a 16-QAM constellation for all waveforms with $N = 128$ and $L=4$. The PS filter is the RRC filter with a roll-off factor of $\alpha = 0.35$. In the case of AFDM, $2Nc_1=8$. The simulated DPAFs of random waveforms are calculated according to (\ref{eq:DPAF_k}), and all simulation results are attained by averaging over 10000 random realizations.

\subsection{DPAFs without PS}

\begin{figure}[!htbp]
	\centering
	\subfigure[Delay Cut]{
		\includegraphics[width=2.3in]{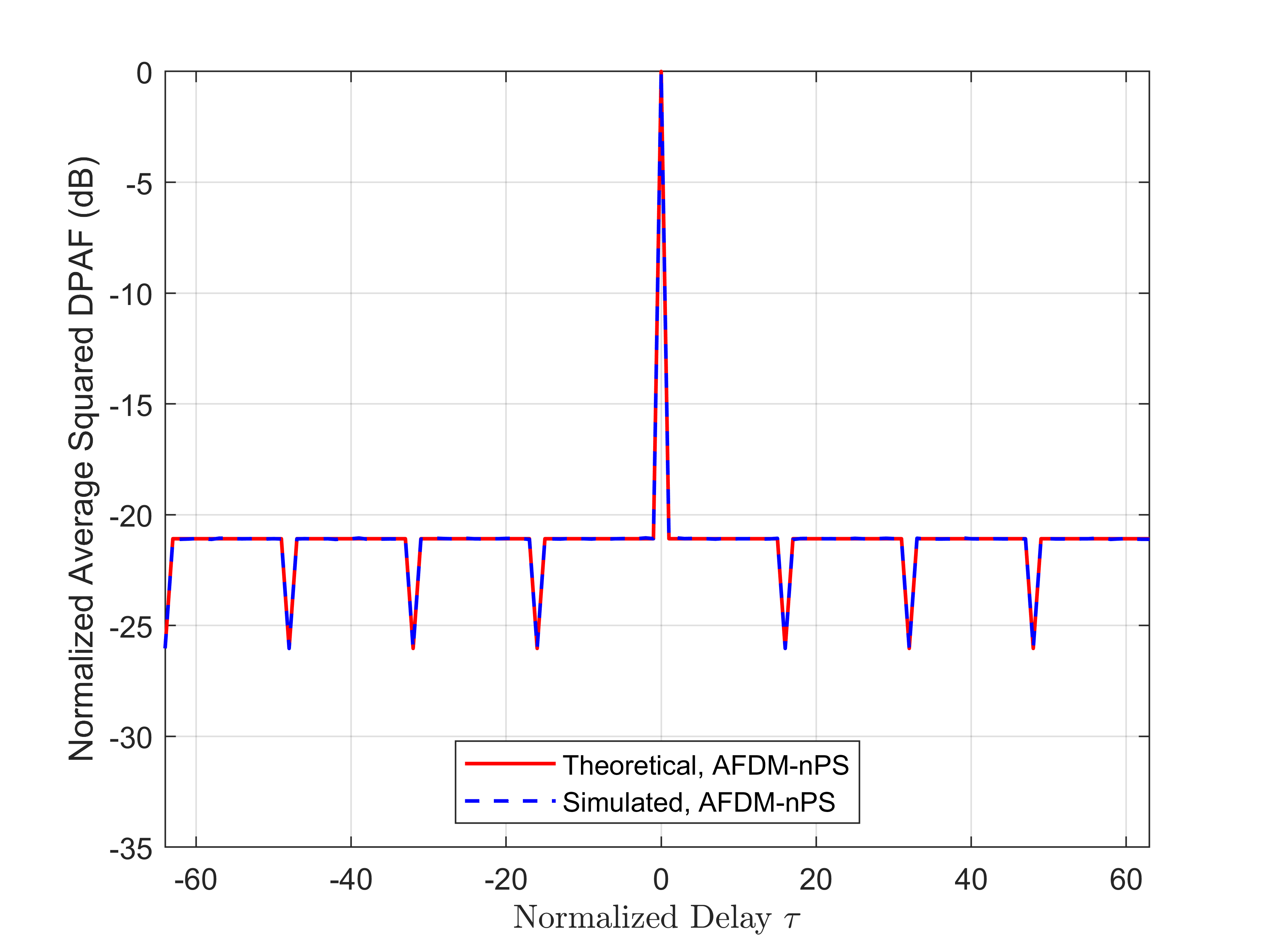}
	}
	\subfigure[Doppler Cut]{
		\includegraphics[width=2.3in]{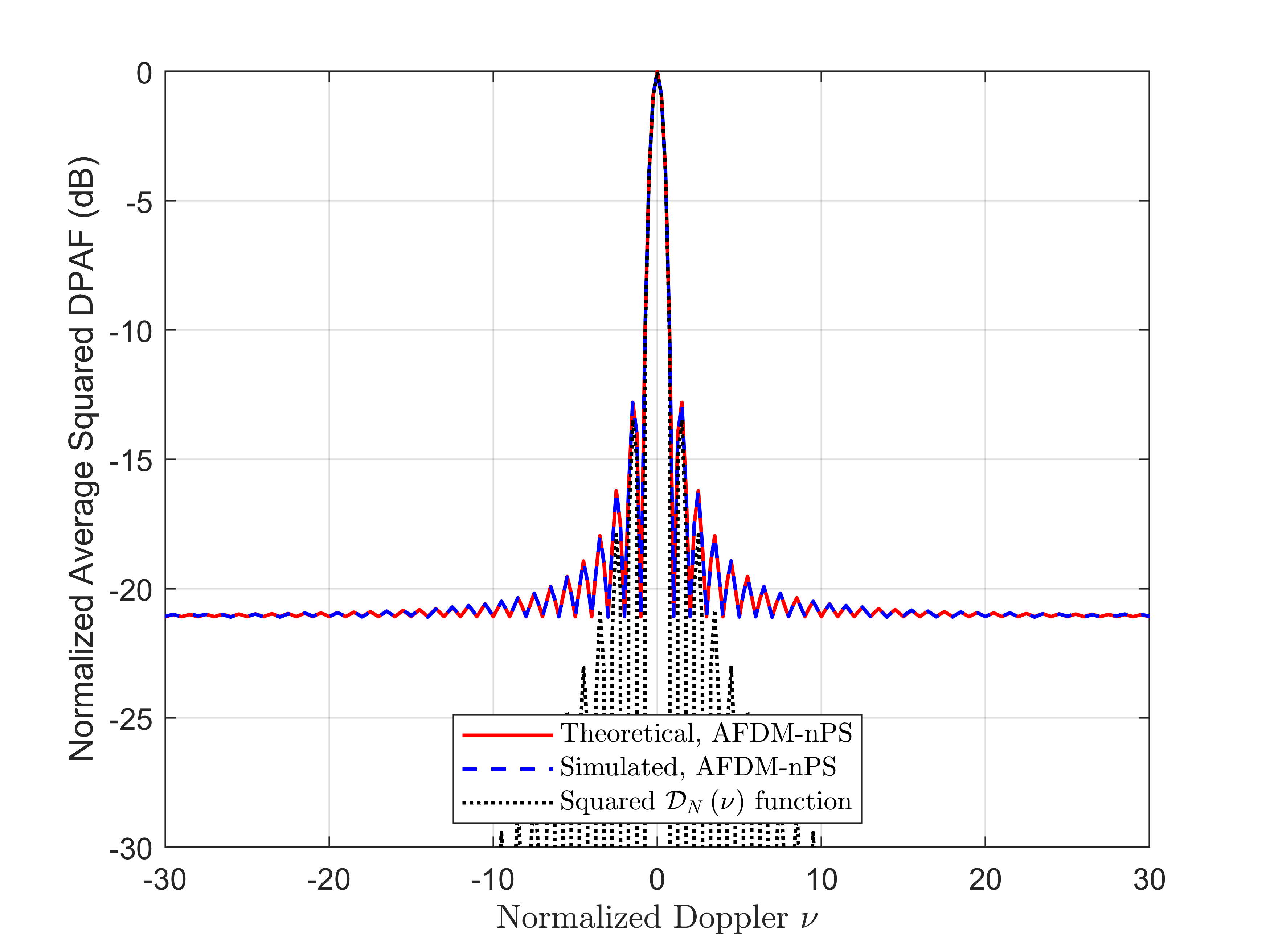}
	}		
	\caption{The theoretical and simulated delay and Doppler Cuts of DPAFs of AFDM without PS for the case of fractional normalized Doppler.  
		\label{fg:sumlated_cut_DPAF_AFDM_noPulse}}
	\vspace*{-5pt} 
\end{figure}

Firstly, we compare simulated and theoretical results of the delay cut and Doppler cut of DPAF of AFDM without PS to verify the effectiveness of our theoretical analysis, as shown in Fig. \ref{fg:sumlated_cut_DPAF_AFDM_noPulse}. ``Theoretical, AFDM-nPS" and ``Simulated, AFDM-nPS" denote the theoretical and simulated results of DPAF of AFDM without PS, respectively, and the theoretical results are obtained from (\ref{eq:DPAF_AFDM_noPulse_fracDop}). On the one hand, the theoretical and simulated delay cuts in Fig. \ref{fg:sumlated_cut_DPAF_AFDM_noPulse}(a) match completely, where the range of normalized delay $\tau$ is in  $\left[-\frac{N}{2},\frac{N}{2}-1\right]$. Moreover, it is shown that there are $2Nc_1-1 = 7$ depressions whose delay indices are $\left(16k,0\right), k\in\left[-4,-1\right]\cup \left[1,3\right]$, which matches the theoretical results shown in Fig. \ref{fg:example_th_DPAF_AFDM}(a) and (b). On the other hand, it can be observed from Fig. \ref{fg:sumlated_cut_DPAF_AFDM_noPulse}(b) that the theoretical Doppler cut is also consistent with the simulated result. Moreover, the mainlobe of Doppler cut matches closely the squared Dirichlet function $\mathcal{{D}}_{N}\left(\nu\right)$ within the normalized Doppler region $\left[-1,1\right]$.

\subsection{DPAFs with PS}

Next, we present the DPAFs of the pulse-shaped AFDM waveform for the integer and fractional normalized Doppler cases in Figs. \ref{fg:simulated_delayCut_DPAF_AFDM_pulsed_intDop} and \ref{fg:simulated_DopplerCut_DPAF_AFDM_pulsed_fracDop}, respectively. ``Theoretical, AFDM-PS'' and ``Simulated, AFDM-PS'' represent the theoretical DPAF of pulse-shaped AFDM in (\ref{eq:AS_DPAF_pulsed_2}) and the simulated DPAF of pulse-shaped AFDM, respectively. ``Squared PACF of pulse'' and ``Squared SSE of pulse'' denote the squared PACF and squared SSE of $\bf{g}$ defined in (\ref{eq:mainlobe_DPAF_pulsed_AFDM}) and (\ref{eq:DopplerCut_DPAF_pulsed_1}), respectively. The simulated average delay cut and Doppler cut of squared DPAFs of pulse-shaped AFDM for the case of integer Doppler are illustrated in Fig. \ref{fg:simulated_delayCut_DPAF_AFDM_pulsed_intDop}. It is shown that the simulated results are completely consistent with the theoretical results for both delay cut and Doppler cut, which demonstrates the effectiveness of our derived closed-form expressions of DPAF. 

\begin{figure}[!htbp]
	\centering
	\subfigure[Simulated Delay Cut]{
		\includegraphics[width=2.3in]{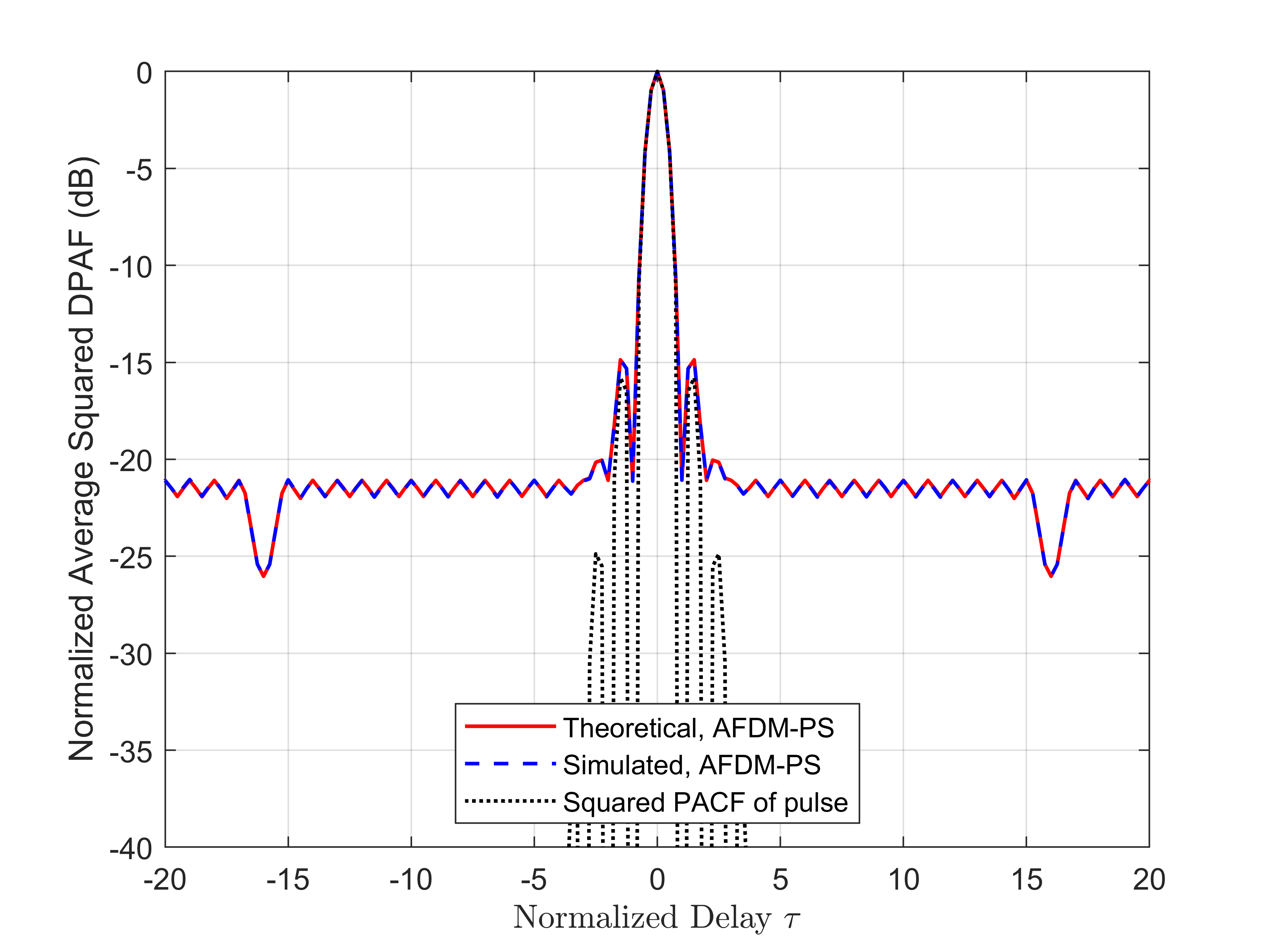}
	}
	\subfigure[Simulated Doppler Cut]{
		\includegraphics[width=2.3in]{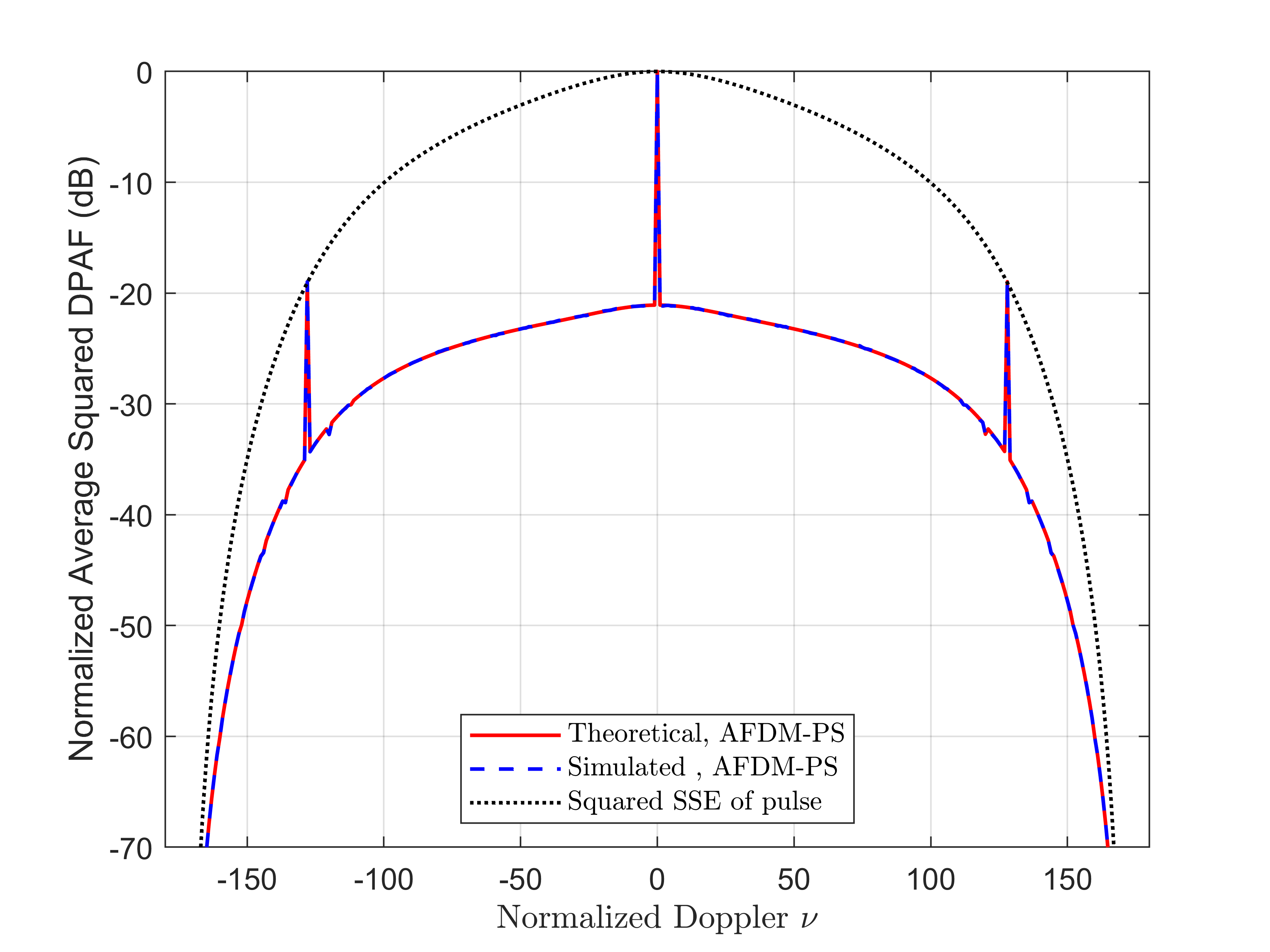}
	}
	\caption{The simulated delay cut and Doppler cut of average squared DPAFs of pulse-shaped AFDM for the case of integer normalized Doppler.  
		\label{fg:simulated_delayCut_DPAF_AFDM_pulsed_intDop}}
	\vspace*{-5pt} 
\end{figure}

Furthermore, the robustness of our derived closed-form expression of DPAF of pulse-shaped AFDM for the case of fractional normalized Doppler is also verified in Fig. \ref{fg:simulated_DopplerCut_DPAF_AFDM_pulsed_fracDop}. The line ``Theoretical, AFDM-PS'' is obtained by substituting fractional normalized Doppler into (\ref{eq:AS_DPAF_pulsed_2}). It is shown that the derived theoretical DPAF of pulse-shaped AFDM is still close to the simulated result. 

\begin{figure}[!htbp]
	\centering
	{
		\includegraphics[width=2.3in]{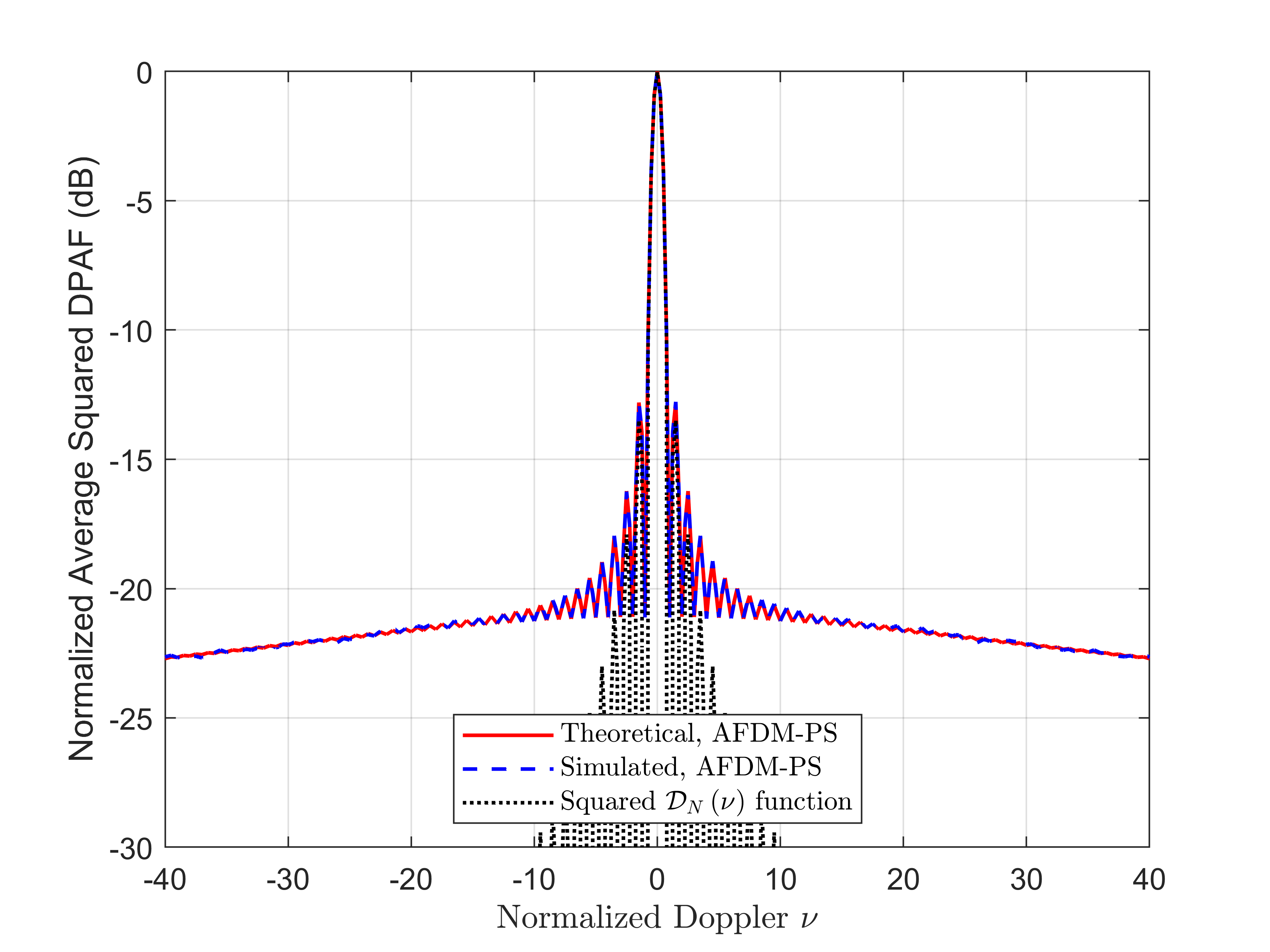}
	}	
	\caption{The simulated Doppler cut of DPAF of pulse-shaped AFDM for the case of fractional normalized Doppler.  
		\label{fg:simulated_DopplerCut_DPAF_AFDM_pulsed_fracDop}}
\end{figure}

Finally, we examine the velocity estimation performance of pulse-shaped AFDM-ISAC and OFDM-ISAC waveforms under a two-target scenario, where strong and weak targets are located at 156.25 m and 937.50 m, respectively, with the same velocity of 100 m/s. The amplitude of the strong target is 21 dB higher than that of the weak target. $c_1 = 0$ for OFDM, and $c_1$ is set according to Corollary \ref{cor:c1_design} for AFDM ($2Nc_1=2$). $N_{\rm sym}=50$ AFDM symbols are processed by noncoherent integration. The maximum likelihood (ML) estimator is used to estimate the velocity for both waveforms. The carrier frequency $f_c = 24$ GHz, and the subcarrier spacing $\Delta _f = 15$ kHz with a symbol duration of $T = 66.67 \mu$s. The root mean squared errors (RMSEs) of velocity estimation of the weak target are illustrated in Fig. \ref{fg:RMSE_velocityEst_OFDM_AFDM_v6}. It is shown that the designed AFDM-ISAC waveform shows 75\% improvement at SNR = 0 dB compared to that of the OFDM-ISAC waveform. 

\begin{figure}[!htbp]
	\centering
	{
		\includegraphics[width=2.3in]{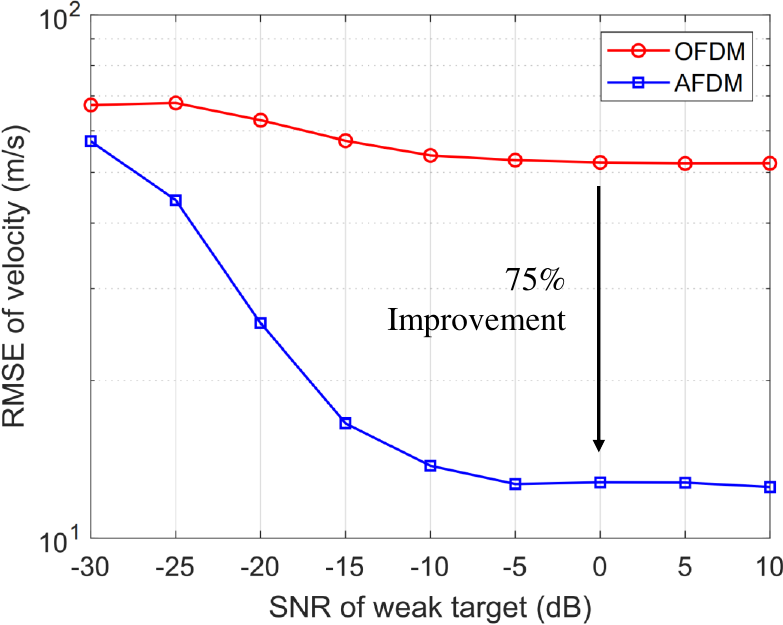}
	}	
	\caption{The RMSE of velocity estimation of the weak target for pulse-shapsed AFDM and OFDM waveforms.  
		\label{fg:RMSE_velocityEst_OFDM_AFDM_v6}}
\end{figure}

\section{Conclusion}

This paper comprehensively analyzed closed-form expressions for the average squared DPAFs of AFDM, OFDM, and OCDM waveforms in both cases with and without PS. We demonstrated that AFDM can control the positions of depressions of DPAF by appropriately adjusting the parameter $c_1$ and proposed a DPAF-inspired design guideline for AFDM parameter $c_1$. Numerical evaluations validated our theoretical findings and proposed design methodology.

\section*{APPENDIX}

\appendices

\section{Proof of Proposition \ref{prop:AS_DPAF_noPS}}\label{proof_prop_AS_DPAF}	
According to Appendix A of our early conference version in \cite{ni2025AFDM}, $\bar s _m$ meets the same assumptions as the symbol $s _m$, and the conclusion in Lemma \ref{Lemma:aver_symbol} still holds for $\bar s _m$, i.e.,
\begin{align}\label{eq:aver_symbol_s_bar}
&\mathbb{E}\left\{ {{{\bar s}_m^{*}} \bar s_{m'}^{} \bar s_n^{} {{\bar s}_{n'}^*}} \right\} \nonumber \\
& = \delta \left( {m - m'} \right)\delta \left( {n - n'} \right)  + \delta \left[ {m - n} \right]\delta \left( {m' - n'} \right) \nonumber \\
& \quad + \left( {{\mu _4} - 2} \right)\delta \left( {m - m'} \right)\delta \left( {m - n} \right)\delta \left( {m - n'} \right),
\end{align}
where $m,m',n,n' \in \left[0,N-1\right]$. Then, substituting (\ref{eq:aver_symbol_s_bar}) into (\ref{eq:DPAF_AFDM_1}), the average squared DPAF of AFDM without PS can be given by (\ref{eq:DPAF_AFDM_noPulse_fracDop_2}) at the bottom of this page. The proof is completed.

\begin{figure*}[b]
	\vspace*{-10pt} 
	\hrule
	\small
	\begin{flalign} \label{eq:DPAF_AFDM_noPulse_fracDop_2}
	&\ \mathbb{E}\left\{ {{{\left| {\chi _{\text{AFDM}} \left( {\tau ,\nu} \right)} \right|}^2}} \right\} = \frac{1}{{{N^2}}}\frac{{{e^{j2\pi \left( {2N{c_1}\tau  - \nu } \right)}} - 1}}{{{e^{j\frac{{2\pi }}{N}\left( {2N{c_1}\tau  - \nu } \right)}} - 1}}\frac{{{e^{ - j2\pi \left( {2N{c_1}\tau  - \nu } \right)}} - 1}}{{{e^{ - j\frac{{2\pi }}{N}\left( {2N{c_1}\tau  - \nu } \right)}} - 1}}\sum\limits_{m = 0}^{N - 1} {{e^{ - j\frac{{2\pi }}{N}m\tau}}\sum\limits_{n = 0}^{N - 1} {{e^{j\frac{{2\pi }}{N}n\tau}}} } &\nonumber\\
	&\ + \frac{{\left( {{\mu _4} - 2} \right)}}{{{N^2}}}\sum\limits_{m = 0}^{N - 1} {\left\{ {\frac{{{e^{j2\pi \left( {2N{c_1}\tau  - \nu } \right)}} - 1}}{{{e^{j\frac{{2\pi }}{N}\left( {2N{c_1}\tau  - \nu } \right)}} - 1}}\frac{{{e^{ - j2\pi \left( {2N{c_1}\tau  - \nu } \right)}} - 1}}{{{e^{ - j\frac{{2\pi }}{N}\left( {2N{c_1}\tau  + \nu } \right)}} - 1}}} \right\}}  + \frac{1}{{{N^2}}}\sum\limits_{m = 0}^{N - 1} {\sum\limits_{m' = 0}^{N - 1} {\left\{ {\frac{{{e^{ - j2\pi \left( {m' - m - 2N{c_1}\tau  + \nu } \right)}} - 1}}{{{e^{ - j\frac{{2\pi }}{N}\left( {m' - m - 2N{c_1}\tau  + \nu } \right)}} - 1}}\frac{{{e^{j2\pi \left( {m' - m - 2N{c_1}\tau  + \nu } \right)}} - 1}}{{{e^{j\frac{{2\pi }}{N}\left( {m' - m - 2N{c_1}\tau  + \nu } \right)}} - 1}}} \right\}} }  &\nonumber\\
	&\	=  \frac{1}{{{N^2}}}{\left| {\frac{{\sin \left[ {\pi \left( {2N{c_1}\tau  - \nu} \right)} \right]}}{{\sin \left[ {{{\pi \left( {2N{c_1}\tau  - \nu} \right)} \mathord{\left/
								{\vphantom {{\pi \left( {2N{c_1}\tau  - \nu} \right)} N}} \right.
								\kern-\nulldelimiterspace} N}} \right]}}} \right|^2}{\left| {\frac{{\sin \left( {\pi \tau} \right)}}{{\sin \left( {{{\pi \tau} \mathord{\left/
								{\vphantom {{\pi \tau} N}} \right.
								\kern-\nulldelimiterspace} N}} \right)}}} \right|^2} 	+ \left( {{\mu _4} - 2} \right)\frac{1}{N}{\left| {\frac{{\sin \left[ {\pi \left( {2N{c_1}\tau  - \nu} \right)} \right]}}{{\sin \left[ {{{\pi \left( {2N{c_1}\tau  - \nu} \right)} \mathord{\left/
								{\vphantom {{\pi \left( {2N{c_1}\tau  - \nu} \right)} N}} \right.
								\kern-\nulldelimiterspace} N}} \right]}}} \right|^2} 	+ \frac{1}{N}\sum\limits_{m = 0}^{N - 1} {{{\left| {\frac{{\sin \left[ {\pi \left( {m + 2N{c_1}\tau  - \nu} \right)} \right]}}{{\sin \left[ {{{\pi \left( {m + 2N{c_1}\tau  - \nu} \right)} \mathord{\left/
										{\vphantom {{\pi \left( {m + 2N{c_1}\tau  - \nu} \right)} N}} \right.
										\kern-\nulldelimiterspace} N}} \right]}}} \right|}^2}} .&
	\end{flalign} \normalsize
	\vspace*{-10pt} 
\end{figure*}		

\section{Proof of Proposition \ref{prop:solution}}\label{proof_prop_solution}	
Since depressions are only located at the indices that meet ${{\left\langle { 2N{c_1}\tau - {\nu} } \right\rangle }_N} = 0$, we start from analyzing the solutions to ${{\left\langle { 2N{c_1}\tau - {\nu}} \right\rangle }_N} = 0$, where $\tau,\nu \in \left[0,N-1\right]$. When $2Nc_1$ is an integer, $2Nc_1\tau = \beta N + b$, where $\beta \in \mathbb{N}$, and $b\in \left[0,N-1\right]$. Hence, for a fixed $\tau \in \left[0,N-1\right]$, the solution to ${{\left\langle { 2N{c_1}\tau - {\nu}} \right\rangle }_N} = 0$ can be given by $\nu = {{\left\langle {2N{c_1}\tau} \right\rangle }_N} = {{\left\langle {b} \right\rangle }_N}$. As $\nu, b \in \left[0,N-1\right]$, we can get that there is only one $\nu = {{\left\langle {2N{c_1}\tau} \right\rangle }_N}$ to make ${{\left\langle { 2N{c_1}\tau - {\nu}} \right\rangle }_N} = 0$. Therefore, ${{\left\langle { 2N{c_1}\tau - {\nu}} \right\rangle }_N} = 0$ has $N$ solutions, which are $\left(\tau,{{\left\langle {2N{c_1}\tau} \right\rangle }_N}\right), \tau \in \left[0,N-1\right]$. Since the index of $\left(0,0\right)$ corresponds to the mainlobe, the remaining $N-1$ solutions correspond to the indices of depressions, that is, there are $N-1$ depressions whose indices are $\left(\tau,{{\left\langle {2N{c_1}\tau} \right\rangle }_N}\right), \tau \in \left[1,N-1\right]$. Proposition \ref{prop:solution} is proved.  

\section{Proof of Corollary \ref{cor:c1_design}}\label{proof:c1_design}

The gap of normalized delay and the gap of normalized Doppler between the strong target and the weak target are equal to
\begin{align}
{\Delta _\tau } &= {{2\left( {{r_w} - {r_s}} \right){f_s}} \mathord{\left/
		{\vphantom {{2\left( {{r_w} - {r_s}} \right){f_s}} c}} \right.
		\kern-\nulldelimiterspace} c},\\
{\Delta _\nu } &= {{2\left( {{v_w} - {v_s}} \right){f_c}N} \mathord{\left/
		{\vphantom {{2\left( {{v_w} - {v_s}} \right){f_c}N} {\left( {c{f_s}} \right)}}} \right.
		\kern-\nulldelimiterspace} {\left( {c{f_s}} \right)}}.
\end{align}
To avoide that the weak target is positioned in the depression of sidelobes of the strong target, it should meet that $\left( {{\Delta _\tau },{{\left\langle {2N{c_1}{\Delta _\tau }} \right\rangle }_N}} \right) \ne \left( {{\Delta _\tau },{\Delta _\nu }} \right)$, that is, ${\left\langle {2N{c_1}{\Delta _\tau }} \right\rangle _N} \ne {\Delta _\nu }$. After mathematical derivation, we can get (\ref{eq:c1_design}). Corollary \ref{cor:c1_design} is proved.

\section{Proof of Proposition \ref{prop:AS_DPAF_pulsed}}\label{proof:AS_DPAF_pulsed}

Substituting (\ref{eq:aver_symbol_s_bar}) into (\ref{eq:AS_DPAF_pulsed_1}), the squared DPAF of pulse-shaped AFDM is rewritten as (\ref{eq:AS_DPAF_pulsed_3}) at the top of next page, where $\mathcal{A}_1$, $\mathcal{A}_2$, $\mathcal{A}_3$ are shown in (\ref{eq:A_1}), (\ref{eq:A_2}) and (\ref{eq:A_3}) at the top of next page through mathematical derivation, respectively. Consequently, substituting (\ref{eq:A_1}), (\ref{eq:A_2}) and (\ref{eq:A_3}) into (\ref{eq:AS_DPAF_pulsed_3}), the average squared DPAF of pulse-shaped AFDM can be expressed as (\ref{eq:AS_DPAF_pulsed_2}). Proposition \ref{prop:AS_DPAF_pulsed} is proved.

\begin{figure*} [htb]
	\vspace*{-10pt} 
	\small
	\begin{flalign}\label{eq:AS_DPAF_pulsed_3}
	&\	\mathbb{E} \left\{ {{{\left| {\chi _{\text{AFDM,PS}} \left( {\tau ,\nu } \right)} \right|}^2}} \right\} & \nonumber\\
	&\ = \frac{1}{{{N^2}}}\sum\limits_{n = 0}^{N - 1} {\sum\limits_{n' = 0}^{N - 1} {\sum\limits_{m = 0}^{N - 1} {\sum\limits_{m' = 0}^{N - 1} {\sum\limits_{q = 0}^{N - 1} {\sum\limits_{q' = 0}^{N - 1} {\sum\limits_{p = 0}^{N - 1} {\sum\limits_{p' = 0}^{N - 1} {\sum\limits_{k = 0}^{NL - 1} {\sum\limits_{k' = 0}^{NL - 1} \Bigg\{ {e^{ - j2\pi {c_1}\left( {{n^2} - {{n'}^2}} \right)}}{e^{ - j\frac{{2\pi }}{N}\left( {mn - m'n'} \right)}}{e^{j2\pi {c_1}\left( {{q^2} - {{q'}^2}} \right)}}{e^{j\frac{{2\pi }}{N}\left( {pq - p'q'} \right)}}   } } } } } } } } } & \nonumber\\	
	&\ \quad\quad\quad\quad\quad\quad\quad\quad\quad\quad\quad \cdot  {g}_{{{\left\langle {k - nL} \right\rangle }_{NL}}} g_{{{\left\langle {k - n'L - \tau } \right\rangle }_{NL}}}  g_{{{\left\langle {k' - qL} \right\rangle }_{NL}}}  {g}_{{{\left\langle {k' - q'L - \tau } \right\rangle }_{NL}}} {e^{  j\frac{{2\pi }}{{NL}}\nu \left( {k - k'} \right)}}  &\nonumber \\
	&\ \quad\quad\quad\quad\quad\quad\quad\quad\quad\quad\quad \cdot { \left[\delta \left( {m - m'} \right)\delta \left( {p - p'} \right)   + \left( {{\mu _4} - 2} \right)\delta \left( {m - m'} \right)\delta \left( {m - p} \right)\delta \left( {m - p'} \right) + \delta \left( {m - p} \right)\delta \left( {m' - p'} \right) \right]} \Bigg\} & \nonumber \\
	&\  = \mathcal{A}_1 + \mathcal{A}_2 + \mathcal{A}_3. &
	\end{flalign}
	\normalsize
	\vspace*{-10pt} 
\end{figure*}
\begin{figure*} [htb]
	\vspace*{-15pt} 
	\small
	\begin{flalign} \label{eq:A_1}
	&\ \mathcal{A}_1 = \frac{1}{{{N^2}}}{\left| {\sum\limits_{n = 0}^{N - 1} {\sum\limits_{n' = 0}^{N - 1} {\sum\limits_{k = 0}^{NL - 1} {{e^{ - j2\pi {c_1}\left( {{n^2} - {{n'}^2}} \right)}}g_{{{\left\langle {k - nL} \right\rangle }_{NL}}}  } } } } {g}_{{{\left\langle {k - n'L - \tau } \right\rangle }_{NL}}} {e^{j\frac{{2\pi }}{{NL}}\nu k}}\sum\limits_{m = 0}^{N - 1} {{e^{ - j\frac{{2\pi }}{N}\left( {n - n'} \right)m}}}   \right|^2}\; & \nonumber\\
	&\ \ \quad = {\left| {\sum\limits_{n = 0}^{N - 1} {{e^{-j\frac{{2\pi }}{N}\nu n}}\sum\limits_{k = 0}^{NL - 1} {g_k {g}_{{{\left\langle {k - \tau } \right\rangle }_{NL}}} {e^{-j\frac{{2\pi }}{{NL}}\nu k}}} } } \right|^2}  = {\left| {\frac{{\sin \left( {\pi \nu} \right)}}{{\sin \left( {{{\pi \nu} \mathord{\left/
								{\vphantom {{\pi \nu} N}} \right.
								\kern-\nulldelimiterspace} N}} \right)}}} \right|^2}{\left| {\sum\limits_{k = 0}^{NL - 1} {g_k {g}_{{{\left\langle {k - \tau } \right\rangle }_{NL}}} {e^{-j\frac{{2\pi }}{{NL}}\nu k}}} } \right|^2} .&
	\end{flalign} \normalsize
	\vspace*{-10pt} 
\end{figure*}
\begin{figure*} [htb]
	\vspace*{-10pt} 
	\small
	\begin{flalign}\label{eq:A_2}
	&\ \mathcal{A}_2 = \frac{{\left( {{\mu _4} - 2} \right)}}{{{N}}}\sum\limits_{n = 0}^{N - 1} {\sum\limits_{q = 0}^{N - 1} {\sum\limits_{q' = 0}^{N - 1} {\sum\limits_{k = 0}^{NL - 1} {\sum\limits_{k' = 0}^{NL - 1} {\Bigg\{ {{e^{ - j2\pi {c_1}\left( {{n^2} - \left\langle {n - q + q'} \right\rangle _N^2} \right)}}} } } } } }  {e^{j2\pi {c_1}\left( {{q^2} - {{q'}^2}} \right)}}{g}_{{{\left\langle {k - nL} \right\rangle }_{NL}}}  g_{{{\left\langle {k - {{\left( {n - q + q'} \right) }}L - \tau } \right\rangle }_{NL}}}  & \nonumber \\
	&\ \quad\quad\quad\quad\quad\quad\quad\quad\quad\quad\quad\quad\quad\quad\quad\quad\quad\quad  { \cdot g_{{{\left\langle {k' - qL} \right\rangle }_{NL}}} {g}_{{{\left\langle {k' - q'L - \tau } \right\rangle }_{NL}}} {e^{  j\frac{{2\pi }}{{NL}}\nu \left( {k - k'} \right)}}} \Bigg\} & \nonumber \\
	&\ \ \quad = \frac{{\left( {{\mu _4} - 2} \right)}}{{{N}}}\sum\limits_{m = 0}^{N - 1} {\Bigg\{ {{{\left| {\sum\limits_{k = 0}^{NL - 1} {g_k {g}_{{{\left\langle {k - \left( {\tau  - mL} \right)} \right\rangle }_{NL}}} {e^{-j\frac{{2\pi }}{{NL}}\nu k}}} } \right|}^2}} }  \left. { \sum\limits_{q = 0}^{N - 1} {{e^{j\frac{{2\pi }}{N}\left( {2N{c_1}m - \nu } \right)q}}} \sum\limits_{n = 0}^{N - 1} {{e^{ - j\frac{{2\pi }}{N}\left( {2N{c_1}m - \nu } \right)n}}} } \right\}&\nonumber \\
	&\ \ \quad = \sum\limits_{m = 0}^{N - 1} {\Bigg\{ \frac{{\left( {{\mu _4} - 2} \right)}}{{{N}}} {{{\left| {\frac{{\sin \left[ {\pi \left( {2N{c_1}m - \nu} \right)} \right]}}{{\sin \left[ {{{\pi \left( {2N{c_1}m - \nu} \right)} \mathord{\left/
											{\vphantom {{\pi \left( {2N{c_1}m - \nu} \right)} N}} \right.
											\kern-\nulldelimiterspace} N}} \right]}}} \right|}^2}} }   {  {{\left| {\sum\limits_{k = 0}^{NL - 1} {g_k g_{{{\left\langle {k - \left( {\tau  - mL} \right)} \right\rangle }_{NL}}} {e^{-j\frac{{2\pi }}{{NL}}\nu k}}} } \right|}^2}} \Bigg\} .&
	\end{flalign} \normalsize
	\vspace*{-10pt} 
\end{figure*}
\begin{figure*} [htb]
	\vspace*{-10pt} 
	\small
	\begin{flalign}\label{eq:A_3}
	&\ \mathcal{A}_3 = \frac{1}{{{N^2}}}\sum\limits_{n = 0}^{N - 1} {\sum\limits_{n' = 0}^{N - 1} {\sum\limits_{q = 0}^{N - 1} {\sum\limits_{p = 0}^{N - 1} {\sum\limits_{k = 0}^{NL - 1} {\mathop \sum \limits_{k' = 0}^{NL - 1} \left\{ {\mathop \sum \limits_{m = 0}^{N - 1} {e^{ - j\frac{{2\pi }}{N}\left( {n - q} \right)m}} } \sum\limits_{m' = 0}^{N - 1} {{e^{j\frac{{2\pi }}{N}\left( {n' - q'} \right)m'}}}{e^{ - j2\pi {c_1}\left( {{n^2} - {{n'}^2}} \right)}}{e^{j2\pi {c_1}\left( {{q^2} - {{q'}^2}} \right)}} \right.} } } } } \;\;  &\nonumber \\
	&\ \quad\quad\quad\quad\quad\quad\quad\quad\quad\quad\quad\quad\quad\quad\quad\quad\quad\quad \cdot  {g}_{{{\left\langle {k - nL} \right\rangle }_{NL}}}  g_{{{\left\langle {k - n'L - \tau } \right\rangle }_{NL}}}  g_{{{\left\langle {k' - qL} \right\rangle }_{NL}}}   {g}_{{{\left\langle {k' - q'L - \tau } \right\rangle }_{NL}}} {e^{ j\frac{{2\pi }}{{NL}}\nu \left( {k - k'} \right)}} \Bigg\} &\nonumber \\
	&\ \ \quad = \sum\limits_{n = 0}^{N - 1} {\sum\limits_{n' = 0}^{N - 1} {{{\left| {\sum\limits_{k = 0}^{NL - 1} {g_k  {g}_{{{\left\langle {k - \left[ {\tau  - \left( {n - n'} \right)L} \right]} \right\rangle }_{NL}}}  {e^{-j\frac{{2\pi }}{{NL}}\nu k}}} } \right|}^2}} }  =  \sum\limits_{m = 0}^{N - 1} { \left\{ N {{\left| {\sum\limits_{k = 0}^{NL - 1} {g_k  {g}_{{{\left\langle {k - \left[ {\tau  - mL} \right]} \right\rangle }_{NL}}} {e^{-j\frac{{2\pi }}{{NL}}\nu k}}} } \right|}^2} \right\} } .
	\end{flalign} \normalsize
	\hrule
	\vspace*{-10pt} 
\end{figure*}

\small
\bibliographystyle{IEEEbib}
\bibliography{IEEEabrv,IEEE_JRCJ_ref}

\end{document}